\documentclass[conference]{IEEEtran}
\usepackage{stfloats}
\usepackage{amsfonts}
\usepackage{amssymb}
\usepackage{amsthm}
\usepackage{cite}
\usepackage[cmex10]{amsmath}
\usepackage{float}
\usepackage{color}
\usepackage{stfloats,fancyhdr}
\usepackage{amsmath,bm}
\usepackage{algorithm}
\usepackage{algorithmic}
\usepackage{multirow}
\usepackage{changepage}
\usepackage[normalem]{ulem}
\usepackage{amsthm}
\usepackage{balance}

\newtheorem{theorem}{Theorem}
\newtheorem{lemma}{Lemma}
\newtheorem{proposition}{Proposition}
\newtheorem{corollary}{Corollary}
\newtheorem{definition}{Definition}

\newtheorem{remark}{Remark}

\ifCLASSINFOpdf
\usepackage[pdftex]{graphicx}
\DeclareGraphicsExtensions{.pdf,.jpeg,.png}
\else
\usepackage[dvips]{graphicx}
\DeclareGraphicsExtensions{.eps}
\fi

\usepackage{subfigure}
\usepackage{fancybox,dashbox}
\usepackage{authblk}

\usepackage{bbm}

\theoremstyle{proof}
\newtheorem{exmp}{Example}%[section]

\begin{document}
%	\parbox[position]{width}{text}
	\title{Optimal Downlink-Uplink Scheduling of Wireless Networked Control for Industrial IoT}
\author{Kang Huang, Wanchun Liu$^\dagger$, Yonghui Li, Branka Vucetic, and Andrey Savkin\\
%	School of Electrical and Information Engineering, The University of 
%	Sydney, Australia\\
%	Emails:	\{kang.huang,\ wanchun.liu,\ yonghui.li,\ branka.vucetic\}@sydney.edu.au. 
}	
	
\maketitle

\begin{abstract}
\let\thefootnote\relax\footnote{
K. Huang, W. Liu, Y. Li and B. Vucetic are with School of Electrical and Information Engineering, The University of 
Sydney, Australia. 
Emails:	\{kang.huang,\ wanchun.liu,\ yonghui.li,\ branka.vucetic\}@sydney.edu.au. 
A. Savkin is with School of Electrical Engineering and Telecommunications,  University of New South Wales, Australia.
Email: a.savkin@unsw.edu.au. (\emph{W. Liu is the corresponding author.})\\	
		Part of the paper has been accepted by Proc. IEEE Globecom 2019~\cite{GC}.
}
This paper considers a wireless networked control system (WNCS) consisting of a dynamic system to be controlled (i.e., a plant), a sensor, an actuator and a remote controller for mission-critical Industrial Internet of Things (IIoT) applications.
%In general, the sensor observes and reports the states of the plant to the remote controller through an uplink channel; the controller has to estimate the current states of the plant due to the sensor's transmission delay and the possible detection errors, and generates and sends control command to the actuator through a downlink channel.
A WNCS has two types of wireless transmissions, i.e., the sensor's measurement transmission to the controller and the controller's command transmission to the actuator.
%, and the packets carrying plant-state information and control commands can be delayed or erroneously detected.
In the literature of WNCSs, the controllers are commonly assumed to work in a full-duplex mode by default, i.e., being able to simultaneously receive the sensor's information and transmit its own command to the actuator.
In this work, we consider a practical half-duplex controller, which introduces a novel \emph{transmission-scheduling problem} for WNCSs. 
A frequent scheduling of sensor's transmission results in a better estimation of plant states at the controller and thus a higher quality of control command, but it leads to a less frequent/timely control of the plant.
Therefore, considering the overall control performance of the plant in terms of its average cost function, there exists a fundamental tradeoff between the sensor's and the controller's transmissions.
We formulate a
new problem to optimize the transmission-scheduling policy %depending on both the current estimation quality of the controller and the current cost function of the plant, 
for minimizing the long-term average cost function. 
We derive the necessary and sufficient condition of the existence
of a stationary and deterministic optimal policy that results in a bounded average cost in terms of the transmission reliabilities of the sensor-to-controller and controller-to-actuator channels. Also, we derive an easy-to-compute suboptimal policy, which notably reduces the average cost of the plant compared to a naive alternative-scheduling policy.
\end{abstract}

\section{Introduction}
Driven by recent development of mission-critical Industrial Internet of Things (IIoT) applications~\cite{ParkSurvey,perera2015emerging,wollschlaeger2017future} and significant advances in wireless communications, networking, computing, sensing and control~\cite{Schulz,bello2014intelligent,zhang2017analysis,shi2016edge}, wireless networked control systems (WNCSs) have recently emerged as a promising technology to enable reliable and remote control of industrial control systems. 
They have a wide range of applications in factory automation, process automation, smart grid, tactile Internet and  intelligent transportation systems~\cite{lyu2018control,gil2013dealing,wang2014oldstyle,smartgrid,transportation}. 
%This technology is driven by mission-critical Industrial Internet of Things (IIoT) applications~\cite{ParkSurvey} and boosted by recent advances in wireless communications,
%networking, computing, sensing and control~\cite{Schulz}. 
Essentially, a WNCS is a spatially
distributed control system consisting of a plant with dynamic states, a set of sensors, a remote controller, and a set of actuators.
%Essentially, a wireless networked control system is a spatially
%distributed control system consisting of a multi-state dynamic plant, a set of sensors that measure and report the plant state, a
%remote controller that collects the sensors' measurements and
%generates and transmits control commands, and a set of actuators that detect and apply the control commands to the plant.

A WNCS has two types of wireless transmissions, i.e., the sensor's measurement transmission to the controller and the controller's command transmission to the actuator. 
The packets carrying plant-state information and control commands can be lost, delayed or corrupted during their transmissions.
Most of existing research in WNCS adopted a \emph{separate design approach}, i.e., either focusing on remote plant-state estimation or remote plant-state control through wireless channels.
%most of the research of WNCSs can be classified in two areas: remote state estimation and remote control.
In \cite{schenato2008optimal} and \cite{yang2015deterministic}, the optimal policies of remote plant-state estimation with a single and multiple  sensors' measurements were proposed, respectively.
Some advanced remote plant-state control methods were investigated to overcome the effects of transmission delay \cite{liu2007networked} and detection errors \cite{demirel2017trade,mishra2018stabilizing}.

The fundamental \emph{co-design} problem of a WNCS in terms of the optimal remote estimation and control were tackled in \cite{schenato2007foundations}. 
%where a sensor measures the plant states and sends its measurements to the controller through a  packet-dropping sensor-controller channel, and the controller estimates the current plant states, and based on which it generates and sends control command to the actuator through a packet-dropping controller-actuator channel.
Specifically, the controller was ideally assumed to work in a \emph{full-duplex} (FD) mode that can simultaneously receive the sensor's packet and transmit its control packet by default. The scheduling of sensor's and controller's transmissions has rarely been considered in the area of WNCSs, while transmission scheduling is actually an important issue for practical wireless communication systems~\cite{mag1,mag2,Yang}.
Moreover, although an FD system can improve the spectrum efficiency, it faces challenges of balancing between the performance of self-interference cancellation, device cost and power consumption, and may not be feasible in practical systems~\cite{inband}.

%From a wireless communication perspective, there are two kinds of full-duplex schemes: the in-band full duplexing (IBDF) and the conventional out-of-band full duplexing, i.e., frequency division duplex (FDD). Although IBDF is a promising technique to improve the spectrum efficiency, a practical IBFD device faces the challenge of balancing the performance of self-interference cancellation with power consumption and device cost and power consumption, and may be far from commercialization.
%Moreover, the conventional \emph{half-duplex} scheme, time division duplexing (TDD), which is the sole duplexing mode in WiFi standards, has also been considered as a more attractive duplexing scheme than FDD in 5G dense deployment environment for higher spectrum flexibility, lower device cost, lower spectrum cost and less overhead in channel estimation~\cite{lahetkangas2014achieving}.
%and flexible uplink-downlink resources allocation

In this paper, we focus on the design of a WNCS using a practical \emph{half-duplex} (HD) controller,
which naturally introduces a fundamental transmission-scheduling problem, i.e., to schedule the sensor's measurement transmission to the controller or the controller's command transmission to the actuator.
A frequent schedule of the sensor's transmission results in a better estimation of plant states and thus a higher quality of the control command. On the other side, a frequent schedule of controller's transmission leads to a more timely plant control.
Thus, considering the overall control performance of plant's states, e.g., the average cost function of the plant, there exists a fundamental tradeoff between the sensor's and the controller's transmission.
We propose a tractable framework to model this problem and enable the optimal design of the WNCS.
%The controller decision schedule problem is first discussed in \cite{imer2006measure} for wired communication system, e.g. control area networks (CAN).
%In this paper, we determine the optimal sense and control policy for wireless networked control system with a half-duplex controller.
The main contributions of the paper are summarized as follows:
\begin{itemize}
\item We propose a WNCS with an HD controller, where the controller schedules the sensor's measurement transmission and its own control-command transmission depending on both the estimation
quality of current plant states and the current cost function of the plant.
\item We formulate a problem to optimally design the transmission-scheduling policy to optimize the long-term control performance of the WNCS in terms of the average cost function for both the one-step and $v$-step controllable plants. 
As the long-term average cost of the plant may not be bounded with high transmission-error probabilities leading to an unstable situation, in the static channel scenario, we derive a necessary and sufficient condition in terms of 
the transmission reliabilities of the sensor-controller and controller-actuator channels and the plant
parameters to ensure the existence of an optimal policy that stabilizes the plant.
In the fading channel scenario, we derive a necessary condition and a sufficient condition in terms of the uplink and downlink channel qualities, under which the optimal transmission scheduling policy exits.
\item We also derive a suboptimal policy with a low computation complexity. The numerical results
show that the suboptimal policy provides an average cost close to the optimal policy, and significantly outperforms the benchmark policy, i.e., scheduling the sensor's and the controller's transmissions alternatively.
\end{itemize}

{The remainder of the paper is organized as follows: In Section~\ref{sec:sys}, we introduce a WNCS with an HD controller. In Section~\ref{sec:ana}, we analyze the estimation-error covariance and the plant-state covariance of the WNCS and formulate an uplink-downlink transmission-scheduling problem.
In Sections~\ref{sec:one} and~\ref{sec:v}, we analyze and solve the transmission-scheduling problem for one-step and multi-step controllable WNCSs, respectively, in the static channel scenario. 
In Section~\ref{sec:fading}, we extend the design to the fading channel scenario.
Section~\ref{sec:num} numerically evaluates the performance
of WNCSs with different transmission-scheduling policies. Finally, Section~\ref{sec:con}
concludes the paper.}

Notations: $\mathbbm{1}(\cdot)$ is the indicator function. $\rho(\mathbf{A})$ denotes the spectral radius of the square matrix~$\mathbf{A}$. $(\cdot)^\top$ is the matrix-transpose operator. $\mathbb{N}$ is the set of positive integers.

\section{System Model}\label{sec:sys}
We consider a discrete-time
WNCS consisting of a dynamic plant with multiple states,
a wireless sensor, an actuator, a remote controller, as illustrated in Fig.~\ref{system_figure}.
In general, the sensor measures the states of the plant and sends the measurements to the remote controller through a wireless uplink (i.e., \emph{sensor-controller}) channel.
The controller generates control commands based on the sensor's feedback and sends the commands to the actuator through a wireless downlink (i.e., \emph{controller-actuator}) channel.
The actuator controls the plant using the received control commands.
\begin{figure}
\centering
\includegraphics[scale=0.8]{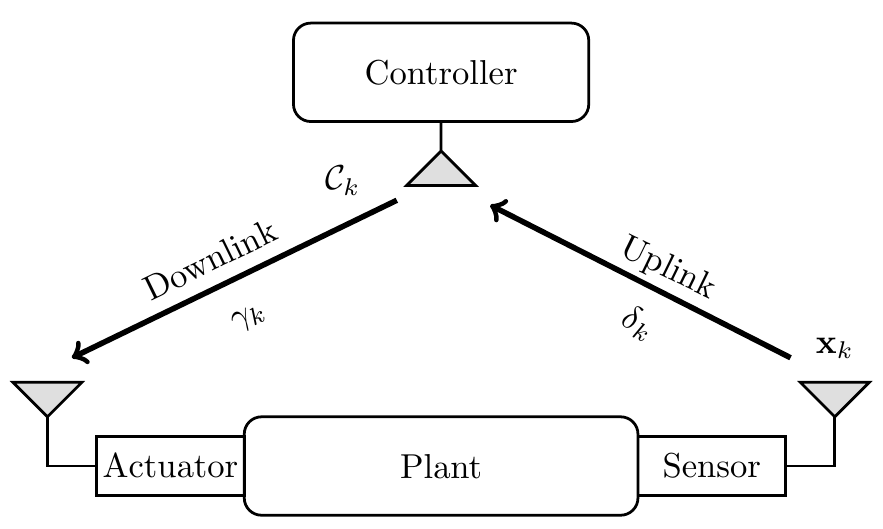}
\vspace{-0.3cm}
\caption{The system architecture.}
\label{system_figure}
\vspace{-0.5cm}
\end{figure}

\subsection{Dynamic Plant}\label{sec:plant}
The plant is a linear time invariant (LTI) discrete-time system modeled as~\cite{demirel2017trade,mishra2018stabilizing,Gatsis}
\begin{equation} \label{system_model}
\mathbf{x}_{k+1} = \mathbf{A}\mathbf{x}_{k} + \mathbf{B}\mathbf{u}_{k} + \mathbf{w}_{k}, \forall k
\end{equation}
where $\mathbf{x}_{k} \in \mathbb{R}^n$ is the plant-state vector at time $k$, $\mathbf{u}_k \in \mathbb{R}^m$ is the control input applied by the actuator and $\mathbf{w}_{k} \in \mathbb{R}^n$ is the plant disturbance independent of $\mathbf{x}_k$ and is a discrete-time zero-mean Gaussian white noise process with the covariance matrix $\mathbf{R} \in \mathbb{R}^{n \times n}$.
$\mathbf{A} \in \mathbb{R}^{n \times n}$ and $\mathbf{B} \in \mathbb{R}^{n \times m}$ are the system-transition matrix and the control-input matrix, respectively, which are constant.
The discrete time step of the system~\eqref{system_model} is $T_0$, i.e., the plant states keep constant during a time slot of $T_0$ and changes slot-by-slot.

We assume that the plant is an unstable system~\cite{schenato2008optimal,demirel2017trade}, i.e., the \emph{spectral radius} of $\mathbf{A}$, $\rho(\mathbf{A})$, is larger than one. In other words, the plant-state vector $\mathbf{x}_{k}$ grows unbounded without the control input, i.e., $\mathbf{u}_k=\mathbf{ 0},\forall k$.
%Thus, our target is to optimally design the controller's policy to drive the dynamic process $\mathbf{x}_k$ close to zero~\cite{demirel2017trade,mishra2018stabilizing,schenato2007foundations},

We consider the long-term \emph{average (quadratic) cost} of the dynamic plant defined as (see e.g. \cite{demirel2017trade,schenato2007foundations})
\begin{equation}  \label{metric}
J = \lim_{K \to \infty}\frac{1}{K}\sum_{k=0}^{K-1} \mathbb{E}\left[\mathbf{x}_k^\top\mathbf{Q}\mathbf{x}_k \right]= \lim_{K \to \infty}\frac{1}{K}\sum_{k=0}^{K-1}\text{Tr}\left(\mathbf{Q}\mathbf{P}_k\right),
\end{equation}
where $\mathbf{Q}$ is a symmetric positive semidefinite matrix, and $\mathbf{P}_k$ is the plant-state covariance defined~as
\begin{equation}
\mathbf{P}_k \triangleq \mathbb{E}\left[ \mathbf{x}_k \mathbf{x}^\top_k \right].
\end{equation}

\begin{definition}[Closed-loop Stability~\cite{demirel2017trade,schenato2007foundations}] \label{def}
	\normalfont
The plant \eqref{system_model} is stabilized by the sequence $\{\mathbf{u}_k\}$, if the average cost function~\eqref{metric} is bounded.
\end{definition}

\subsection{HD Operation of the Controller}
We assume that the controller is an HD device, and thus it can either receive the sensor's measurement or transmit its control command to the actuator at a time.
Let $a_k \in \{1,2\}$ be the controller's transmission-scheduling variable in time slot $k$. The sensor's or the controller's transmission is scheduled in time slot~$k$ if $a_k=1$ or $2$, respectively.

The sensor measures the plant states at the beginning of each time slot. The measurement is assumed to be perfect~\cite{demirel2017trade,mishra2018stabilizing,Gatsis}. We use $\delta_k$ to indicate the successfulness of the sensor's transmission in time slot $k$. Thus, $\delta_k=1$ if the sensor is scheduled to send a packet carrying its measurement to the controller in time slot $k$ (i.e., $a_k=1$) and the transmission is successful, and $\delta_k=0$ otherwise.

The controller generates a control-command-carrying packet at the beginning of each time slot. Similarly, we use $\gamma_k$ to indicate the successfulness of the controller's transmission in time slot $k$. Thus, $\gamma_k=1$ if the controller is scheduled to send the control packet to the actuator in time slot $k$ (i.e., $a_k=2$) and the transmission is successful, and $\gamma_k=0$ otherwise.
We also assume that the controller has a perfect feedback from the actuator indicating the successfulness of the packet detection~\cite{schenato2007foundations}. \emph{Thus, the controller knows whether its control command will be applied or not.}
We assume that the packets in both the sensor-to-controller and controller-to-actuator channels have the same packet length and is less than $T_0$~\cite{schenato2008optimal,schenato2007foundations}.

\subsection{Wireless Channel}
{We consider both the static channel and the fading channel scenarios of the WNCS. The static channel scenario is for the IIoT applications with low mobilities, e.g., process control of chemical and oil-refinery plant, while the fading channel scenario is for high mobility applications, e.g., motion control of automated guided vehicles in warehouses.

For the static channel scenario, we assume that the \emph{packet-error probabilities} of the uplink (sensor-controller) and downlink (controller-actuator) channels are $p_s$ and $p_c$, respectively, which do not change with time, where $p_s,p_c \in \left(0,1\right)$.

For the fading channel scenario, we adopt a practical finite-state Markov channel model, which captures the 
inherent property of practical fading channels for which the channel states change with memories~\cite{Parastoo}. It is assumed that the uplink channel and the downlink channel have $B_s$ and $B_c$ states, respectively, and the packet loss probability of the $i$th channel state of the uplink channel and the $j$th channel state of the downlink channel are $\omega_i$ and $\xi_j$, respectively. The matrices of the channel state transition probabilities of the uplink and downlink channels are given as
\begin{equation} \label{channel_transition1}
\mathbf{{D}}_s \triangleq \begin{bmatrix}
{d}^s_{1,1}  & \cdots &	{d}^s_{B_s,1}\\
\vdots  & \ddots &	\vdots\\
{d}^s_{1,B_s}  & \cdots &	{d}^s_{B_s,B_s}
\end{bmatrix},
\end{equation}
and
\begin{equation} \label{channel_transition2}
\mathbf{{D}}_c \triangleq \begin{bmatrix}
{d}^c_{1,1}  & \cdots &	{d}^c_{B_c,1}\\
\vdots  & \ddots &	\vdots\\
{d}^c_{1,B_c}  & \cdots &	{d}^c_{B_c,B_c}
\end{bmatrix},
\end{equation}
respectively. 
The packet-error probabilities of the uplink and downlink channels at time $k$ are $p_{s,k}$ and $p_{c,k}$, respectively, where $p_{s,k}\in\{\omega_1,\cdots,\omega_{B_s}\}$ and $p_{c,k}\in\{\xi_1,\cdots,\xi_{B_c}\}$.}

\subsection{Optimal Plant-State Estimation}
At the beginning of time slot $(k+1)$,
%that has been scheduled for the transmission of a control command, i.e., $a_{k+1}=2$, 
before generating a proper control command, the controller needs to estimate the current states of the plant, $\mathbf{x}_{k+1}$, using the previously received sensor's measurement and also the implemented control input based on the  dynamic plant model~\eqref{system_model}. The optimal plant-state estimator is given as~\cite{schenato2008optimal}
\begin{equation} \label{estimation}
\hat{\mathbf{x}}_{k+1} = \begin{cases}
\mathbf{A}\mathbf{x}_{k} + \mathbf{B}\mathbf{u}_{k}, &a_k =1,\delta_{k}=1,\\
\mathbf{A}\hat{\mathbf{x}}_{k} + \mathbf{B}\mathbf{u}_{k}, &\text{otherwise}.
\end{cases}
\end{equation}

%To indicate the number of passed time slots from the last successfully received sensor's packet,
%we define the \emph{estimation-quality indicator} of the plant in time slot $k$ as
%\begin{equation}\label{tau}
%\tau_{k+1} = \begin{cases}
%1, &a_k=1, \delta_{k} = 1,\\
%\tau_{k} +1, & \text{otherwise}.
%\end{cases}
%\end{equation}
%
%where $\tau_k$ is the \emph{estimation-quality indicator} of the plant in time slot $k$. Specifically, $\tau_k$ is the number of the time slots passed from the most recent sensor's successful transmission to the current time slot $k$.

\subsection{$v$-Step Predictive Plant-State Control}
%We assume that the plant \eqref{system_model} is completely one-step controllable~\cite{demirel2017trade}, e.g., there exists a matrix $\mathbf{K}\in \mathbb{R}^{m\times n}$ such that  \begin{equation}\label{a+bk}
%\mathbf{A}+\mathbf{B}\mathbf{K}=\mathbf{0}.
%\end{equation}

%in the sense of minimizing the quadratic cost of the plant in~\eqref{metric}.

%
%If $\mathbf{u}^c_k$ is scheduled to be transmitted and is received by the actuator successfully in time slot $k$, the actuator applies it to the plant immediately; otherwise, the actuator does not control the plant~\cite{schenato2007foundations}. Thus, the control input $\mathbf{u}_k$ in~\eqref{system_model} is updated~as
%\begin{equation} \label{uk}
%\mathbf{u}_{k} = \begin{cases}
%\mathbf{u}^c_{k}, &a_k =2,\gamma_k=1,\\
%\mathbf{0}, &\text{otherwise}.
%\end{cases}
%\end{equation}	

As the transmission between the controller and the actuator is unreliable, the actuator may not successfully receive the controller's packet containing the current control command. To provide robustness against packet failures, we consider a predictive control approach~\cite{gupta2006receding}. In general, the controller sends both the current command and the predicted future commands to the actuator at each time. 
If the current command-carrying packet is lost, the actuator will apply the previously received command that was predicted for the current time slot. The details of the predictive control method is given below.

The controller adopts a conventional linear predictive control law~\cite{demirel2017trade}, which generates a sequence of $v$ control commands including one current command and $(v-1)$ predicted future commands in each time slot $k$ as
\begin{equation} \label{C}
\mathcal{C}_k = \big[\mathbf{K}\hat{\mathbf{x}}_k,\underbrace{\mathbf{K}(\mathbf{A}+\mathbf{B}\mathbf{K})\hat{\mathbf{x}}_k,\cdots,\mathbf{K}(\mathbf{A}+\mathbf{B}\mathbf{K})^{v-1}\hat{\mathbf{x}}_k}_{(v-1) \text{ predicted control commands}}\big],
\end{equation}
where the constant $v$ is the length of predictive control, and the constant $\mathbf{K}\in \mathbb{R}^{m\times n}$ is the controller gain, which satisfies the condition that\footnote{If \eqref{rho_ABK} is not satisfied, the plant \eqref{system_model} can never be stabilized even if the uplink and downlink transmissions are always perfect see e.g., \cite{schenato2007foundations,GatsisOppor}.} 
\begin{equation}\label{rho_ABK}
\rho(\mathbf{ A+BK})<1.
\end{equation}
If time slot $k$ is scheduled for the controller's transmission, the controller sends a packet containing $v$ control commands $\mathcal{C}_k$ to the actuator.
Note that in most communication protocols, the minimum packet length is longer than the time duration required for transmitting a single control command~\cite{gupta2006receding}, and thus it is wise to send multiple commands in the one packet without increasing the packet length.

The actuator maintains a command buffer of length $v$, $\mathcal{U}_k\triangleq \left[\mathbf{u}^0_k,\mathbf{u}^1_k,\cdots,\mathbf{u}^{v-1}_k\right]$.
If the current controller's packet is successfully received, the actuator resets the buffer with the received command sequence, otherwise, the buffer shifts one step forward, i.e.,
\begin{equation} \label{U}
\mathcal{U}_k
%\triangleq \left[\mathbf{u}^0_k,\mathbf{u}^1_k,\cdots,\mathbf{u}^{v-1}_k\right]
=
\begin{cases}
\mathcal{C}_k, & a_k=2,\gamma_k=1,\\
\left[\mathbf{u}^1_{k-1},\mathbf{u}^2_{k-1},\cdots,\mathbf{u}^{v-1}_{k-1},\mathbf{0}\right], & \text{otherwise.}
\end{cases}
\end{equation}
The actuator always applies the first command in the buffer to the plant. Thus, the actuator's control input in time slot $k$~is 
\begin{equation} \label{u}
\mathbf{u}_k \triangleq \mathbf{u}^0_k.
\end{equation}

To indicate the number of passed time slots from the last successfully received control packet,
we define the \emph{control-quality indicator} of the plant in time slot $k$ as
\begin{equation}\label{eta}
\eta_{k} = \begin{cases}
1, &a_k=2, \gamma_{k} = 1,\\
\eta_{k-1} +1, & \text{otherwise}.
\end{cases}
\end{equation}
Specifically, $\eta_{k-1}$ is the number of the time slots passed from the most recent controller's successful transmission to the current time slot $k$.

From \eqref{C}, \eqref{U}, \eqref{u} and \eqref{eta}, the control input can be rewritten~as 
\begin{equation} \label{control_command}
\mathbf{u}_{k} = \begin{cases}
\mathbf{K}(\mathbf{A}+\mathbf{B}\mathbf{K})^{\eta_k-1}\hat{\mathbf{x}}_{k+1-\eta_k}, &\mbox{ if } \eta_k \leq v,\\
\mathbf{0}, &\mbox{ if } \eta_k > v.
\end{cases}
\end{equation}

To better explain the intuition behind the predictive control method \eqref{C}, \eqref{U} and \eqref{u}, we give an example below.

\begin{exmp}	
Assume that a sequence of the controller's commands is successfully received in time slot $k$ and the actuator will not receive any further commands in the following $v-1$ time slots. Consider an ideal case that the estimation is accurate in time slot $k$, i.e., $\hat{\mathbf{x}}_k={\mathbf{x}}_k$, and the plant disturbance, $\mathbf{w}_k=\mathbf{0},\forall k$.
Taking \eqref{control_command} into \eqref{system_model}, the plant-state vector at $(k+j), \forall j\leq v$ can be derived as
\begin{equation}\label{example}
\mathbf{x}_{k+j}=(\mathbf{ A+BK})^j \mathbf{x}_{k}.
\end{equation}
Therefore, if the controller gain $\mathbf{K}$ is chosen properly and makes the spectral radius of $(\mathbf{A+BK})$ less than one, each state in $\mathbf{x}_{k}$ can approach zero gradually within the $v$ steps even without receiving any new control packets.
\end{exmp}

In this work, we mainly focus on two types of plants applying the predictive control method as follows.

\emph{\underline{Case 1}:}
The controller gain $\mathbf{K}$ satisfies the condition that
\begin{equation}\label{a+bk}
\mathbf{A}+\mathbf{B}\mathbf{K}=\mathbf{0}.
\end{equation}
This case is named as the \emph{one-step controllable} case~\cite{OREILLY1981363}, since 
%
%In this case, the plant \eqref{system_model} is (ideally) one-step controllable~\cite{OREILLY1981363} by applying the predictive control method \eqref{C}, \eqref{U} and \eqref{u}.
%In other words, 
once a control packet is received successfully, the plant-state vector $\mathbf{x}_{k}$ can be driven to zero in one step in the above mentioned ideal setting, i.e., 
$\mathbf{x}_{k+1} =\mathbf{0} \mathbf{x}_k=\mathbf{0}$ in \eqref{example}.
%
% once a control packet is received successfully in the absence of plant disturbance, i.e., $\mathbf{x}_{k+1}=\mathbf{A}\mathbf{x}_{k} + \mathbf{B}\mathbf{u}_{k}
%=\mathbf{A}\mathbf{x}_{k} + \mathbf{B}(\mathbf{K}\mathbf{x}_{k})=\mathbf{ 0}$.
By taking \eqref{a+bk} into \eqref{C}, the $(v-1)$ predicted commands are all $\mathbf{0}$, thus the controller only needs to send the current control command to the actuator \emph{without any prediction}, and the length of $\mathcal{U}$ and $\mathcal{C}$, $v$, is equal to one.

\emph{\underline{Case 2}:} 
The controller gain $\mathbf{K}$ satisfies the condition that~\cite{OREILLY1981363}
\begin{equation}\label{a+bk2}
(\mathbf{A}+\mathbf{B}\mathbf{K})^v=\mathbf{0},v>1.
\end{equation}
This case is named as the \emph{$v$-step controllable} case~\cite{OREILLY1981363}, since 
%In this case, the plant \eqref{system_model} is (ideally) $v$-step  controllable by applying the predictive control method \eqref{C}, \eqref{U} and \eqref{u}.
%In other words, 
the plant state $\mathbf{x}_{k}$ can be driven to zero in $v$ steps after a successful reception of a control packet in the ideal setting\footnote{Note that the ideal setting here is only for the explanation of the term of ``one-step controllable", while we only consider practical settings in the rest of the paper.}, i.e., $\mathbf{x}_{k+v}=\mathbf{ 0}$ in \eqref{example}.

The other cases not satisfying the conditions \eqref{a+bk} nor \eqref{a+bk2}, will also be discussed in the following section.

%$\eta_k$ is the \emph{control-quality indicator} of the plant in time slot $k$. Specifically, $\eta_{k-1}$ is the number of the time slots passed from the most recent controller's successful transmission to the current time slot $k$.

\section{Analysis of the Downlink-Uplink Scheduling}\label{sec:ana}
As the controller estimates the current plant states and utilizes the estimation to control the future plant states, we analyze the estimation-error covariance and the plant-state covariance in the sequel.

\subsection{Estimation-Error Covariance}
Using \eqref{system_model} and \eqref{estimation}, the estimation error in time slot $(k+1)$ is obtained as
\begin{equation} \label{estimation_error}
\mathbf{e}_{k+1} \triangleq \mathbf{x}_{k+1} - \hat{\mathbf{x}}_{k+1} = \begin{cases}
\mathbf{w}_{k}, &a_k=1, \delta_{k} = 1,\\
\mathbf{A}\mathbf{e}_{k} + \mathbf{w}_{k}, &\text{otherwise}.
\end{cases}
\end{equation}
Thus, we have the updating rule of the estimation-error covariance, $\mathbf{U}_{k} \triangleq \mathbb{E}[\mathbf{e}_k\mathbf{e}_k^\top] $,  as
\begin{equation}\label{e_covariance}
\mathbf{U}_{k+1} \triangleq\mathbb{E}[\mathbf{e}_{k+1}\mathbf{e}_{k+1}^\top] = \begin{cases}
\mathbf{R} &a_k=1, \delta_{k} = 1,\\
\mathbf{A} \mathbf{U}_k \mathbf{A}^\top + \mathbf{R} &\text{otherwise}.
\end{cases}
\end{equation}

We define the \emph{estimation-quality indicator} of the plant in time slot $k$, $\tau_k$, as the number of passed time slots from the last successfully received sensor's packet. Then, the state-updating rule of $\tau_k$ is obtained as
\begin{equation}\label{tau}
\tau_{k+1} = \begin{cases}
1, &a_k=1, \delta_{k} = 1,\\
\tau_{k} +1, & \text{otherwise}.
\end{cases}
\end{equation}

Once a successful sensor's transmission occurs (e.g., there exists $k'$ such that $\mathbf{U}_{k'} = \mathbf{R}$),
from \eqref{e_covariance} and \eqref{tau}, it can be shown that the estimation-error covariance $\mathbf{U}_k,\forall k\geq k'$, is simply a function of the estimation-quality indicator $\tau_k$, i.e.,
\begin{equation} \label{Uk}
\mathbf{U}_k = \mathbf{F}(\tau_k),
\end{equation}
where the function $\mathbf{F}(\cdot)$ is defined~as
\begin{equation} \label{F}
\mathbf{F}(\tau) \triangleq \sum_{i=1}^{\tau}\mathbf{A}^{i-1}\mathbf{R}(\mathbf{A}^\top)^{i-1}, \tau \in \mathbb{N}.
\end{equation}

%Thus, $\mathbf{U}_k$ will only take value from a countable infinity
%set $\{\mathbf{F}(1),\mathbf{F}(2),\mathbf{F}(3),\cdots\}$. 
As we focus on the long-term performance of the system, without loss of generality, we assume that $\mathbf{U}_k \in \{\mathbf{F}(1),\mathbf{F}(2),\mathbf{F}(3),\cdots\}$ for all $k$.
% Specifically, $\tau_k$ is the number of the time slots passed from the most recent sensor's successful transmission to the current time slot $k$, and thus 
From~\eqref{tau} and \eqref{Uk}, the updating rule of $\mathbf{U}_k$ is obtained as
\begin{equation}\label{e_covariance2}
\mathbf{U}_{k+1} = \mathbf{F}(\tau_{k+1})  = \begin{cases}
\mathbf{F}(1) &a_k=1, \delta_{k} = 1,\\
\mathbf{F}(\tau_k+1) &\text{otherwise}.
\end{cases}
\end{equation}

\subsection{Plant-State Covariance of One-Step Controllable Case}	

Taking \eqref{eta} and \eqref{a+bk} into \eqref{control_command}, the control input of the one-step controllable case can be simplified as
%\begin{equation} \label{uc_k}
%\mathbf{u}^c_k = \mathbf{K}\hat{\mathbf{x}}_k
%\end{equation}
%
%As the control command \eqref{uc_k} only take effects on current time step, no extra actuator buffer is required. The actual control input is update based on the current packets arrival as
\begin{equation} \label{uk}
\mathbf{u}_{k} = \begin{cases}
\mathbf{K}\hat{\mathbf{x}}_{k}, &a_k =2,\gamma_k=1,\\
\mathbf{0}, &\text{otherwise}.
\end{cases}
\end{equation}
Substituting \eqref{uk} into \eqref{system_model} and using \eqref{a+bk}, the plant-state vector can be rewritten~as	
\begin{equation} \label{system_evolution}
\mathbf{x}_{k+1} \!=\! \begin{cases}
\mathbf{A}\mathbf{x}_{k} + \mathbf{B}\mathbf{K}\hat{\mathbf{x}}_k + \mathbf{w}_{k} \!=\! \mathbf{A}\mathbf{e}_k + \mathbf{w}_{k}, &\!a_k=2,\gamma_k=1,\\
\mathbf{A}\mathbf{x}_{k} + \mathbf{w}_{k}, &\text{otherwise.}
\end{cases}
\end{equation}
Thus, the plant-state covariance, $\mathbf{P}_k$, has the updating rule as
\begin{equation}\label{P_covariance}
\mathbf{P}_{k+1} \triangleq \mathbb{E}[\mathbf{x}_{k+1}\mathbf{x}_{k+1}^\top] = \begin{cases}
\mathbf{A}\mathbf{U}_k \mathbf{A}^\top +  \mathbf{R} &a_k=2,\gamma_{k} = 1,\\
\mathbf{A}\mathbf{P}_k \mathbf{A}^\top +  \mathbf{R} &\text{otherwise.}
\end{cases}
\end{equation}
From \eqref{F}, \eqref{e_covariance2} and \eqref{P_covariance}, we see that the plant-state covariance $\mathbf{P}_k$ will only take value from the countable infinity
set $\{\mathbf{F}(2),\mathbf{F}(3),\cdots\}$ after a successful controller's transmission. 
Again, as we focus on the long-term performance of the system, we assume that $\mathbf{P}_k \in \{\mathbf{F}(2),\mathbf{F}(3),\cdots\}$ for all $k$, without loss of generality.

By introducing the variable $\phi_k \in \{2,3,\cdots \}$, the plant-state covariance in time slot $k$ can be written as
\begin{equation} \label{Pk}
\mathbf{P}_k = \mathbf{F}(\phi_k),
\end{equation}
where $\phi_k$ is the \emph{state-quality indicator} of the plant in time slot~$k$. Note that the state covariance only depends on the state parameter $\phi_k$.

From \eqref{P_covariance} and \eqref{Uk}, the updating rules of $\mathbf{P}_k$ and $\phi_k$ in \eqref{Pk} are given by, respectively, as
\begin{equation}\label{P_covariance2}
\mathbf{P}_{k+1} = \mathbf{F}(\phi_{k+1})  = \begin{cases}
\mathbf{F}(\tau_k +1) &a_k=2, \gamma_{k} = 1,\\
\mathbf{F}(\phi_k+1) &\text{otherwise},
\end{cases}
\end{equation}
\begin{equation}\label{phi}
\phi_{k+1} = \begin{cases}
\tau_k+1, &a_k=2, \gamma_{k} = 1,\\
\phi_{k} +1, & \text{otherwise}.
\end{cases}
\end{equation}
From \eqref{tau} and \eqref{phi}, it is easy to prove that $\phi_k\geq \tau_k, \forall k$.

\subsection{Plant-State Covariance of $v$-Step Controllable Case}
%Consider that the plant \eqref{system_model} is $v$-step completely controllable, i.e., the controller gain $\mathbf{K}$ satisfies the condition that
%\begin{equation}\label{a+bk2}
%(\mathbf{A}+\mathbf{B}\mathbf{K})^v=\mathbf{0}.
%\end{equation}
%In other words, the plant state $\mathbf{x}_{k}$ can be driven to zero in $v$ steps in \eqref{system_model} after a successful reception of a control packet in the absence of plant disturbance, i.e., $\mathbf{x}_{k+v}=\mathbf{0}$ in \eqref{example}.

Taking \eqref{control_command} into \eqref{system_model}, the plant-state vector is rewritten as
\begin{equation} \label{state1}
\mathbf{x}_{k+1} \!=\! \begin{cases}
\!\mathbf{A}\mathbf{x}_{k} \!+\! \mathbf{B}\mathbf{K}(\mathbf{A}+\mathbf{B}\mathbf{K})^{\eta_k-1}\hat{\mathbf{x}}_{k+1-\eta_k} \!+\! \mathbf{w}_{k}, &\mbox{\!\!\!\!if } \eta_k \!\leq v,\\
\!\mathbf{A}\mathbf{x}_{k} + \mathbf{w}_{k}, &\mbox{\!\!\!\!if } \eta_k \!> v.
\end{cases}
\end{equation}
Using the property \eqref{a+bk2}, we have the state-updating rule as
\begin{equation} \label{v_x}
\mathbf{x}_{k} = \mathbf{A}\mathbf{x}_{k-1} + \mathbf{B}\mathbf{K}(\mathbf{A}+\mathbf{B}\mathbf{K})^{\eta_{k-1}-1}\hat{\mathbf{x}}_{k-\eta_{k-1}} + \mathbf{w}_{k-1}.
\end{equation}
Different from the one-step controllable case in \eqref{system_evolution}, where the current state vector relies on the previous-step estimation, it depends on the state estimation $\eta_{k-1}$ steps ago in the $v$-step controllable case.

Inspired by the one-step controllable case \eqref{Pk}, we aim at deriving the plant-state covariance in terms of a set of state parameters.
First, we define a sequence of variables, $t^i_k$, $i=1,\cdots,v$, where $t^i_k$ is the time-slot index of the $i$th latest successful controller's transmission prior to the current time slot $k$ as illustrated in Fig.~\ref{fig:time}.
\begin{figure}[t]
	\centering
	\includegraphics[scale=0.8]{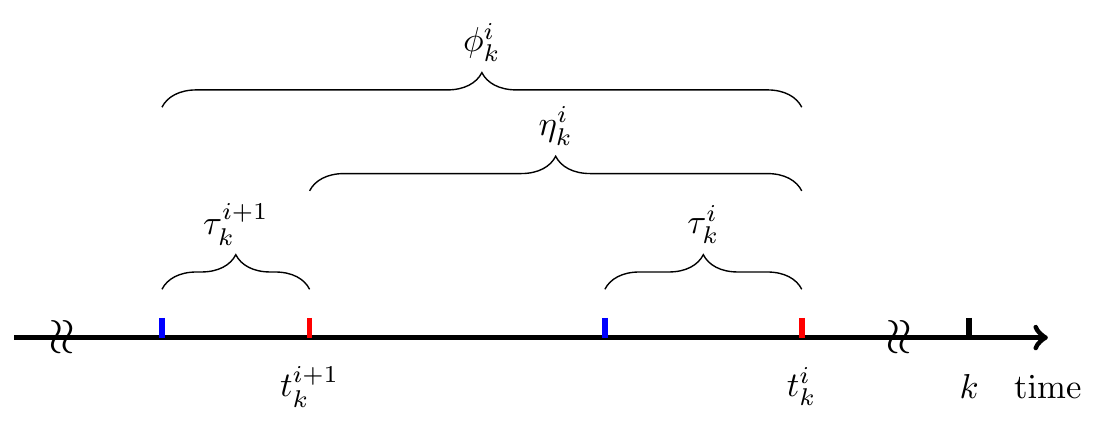}
	\vspace{-0.5cm}	
	\caption{Illustration of the state parameters, where red vertical bars denote successful controller's transmissions and blue vertical bars denote the most recent successful sensor's transmissions prior to the successful controller's transmissions.}
	\label{fig:time}
	\vspace{-0.5cm}
\end{figure}
Then, we define the following state parameters
\begin{align}
\label{tau1}
\tau^i_k &\triangleq \begin{cases}
\tau_{k},&i=0,\\
\tau_{t^i_k},&i=1,2,\cdots v,
\end{cases}
\\
\label{eta1}
\eta^i_k &\triangleq \begin{cases}
\eta_{k-1}=k-t^1_k,&i=0,\\
%\eta^0_{t^i_k} = 
t^i_k-t^{i+1}_k,&i=1,2,\cdots v-1.
\end{cases}
\end{align}
Specifically, $\eta^i_k$ measures the delay between two consecutive controller's successful transmissions; $\tau^i_k$ is the estimation-quality indicator of time slot $t^i_k$.
Last, we define the state parameters $\phi^i_k$ as
\begin{equation}\label{phi1}
\phi^i_k \triangleq \eta^i_k +\tau^{i+1}_k, i= 0,\cdots,v-1.
\end{equation}

Using the state-transition rules of $\eta_k$ and $\tau_k$ in \eqref{eta} and \eqref{tau}, and the definitions \eqref{tau1}, \eqref{eta1} and \eqref{phi1}, the state-transition rules of $\tau^i_k$, $\eta^i_k$ and $\phi^i_k$ can be obtained, respectively,
\begin{align} \label{tau2}
\tau^i_{k+1} &= \begin{cases}
1,&i=0, a_k=1,\delta_k=1,\\
\tau^{0}_{k}+1,& i=0, \text{ otherwise},\\
\tau^{i-1}_{k},&i=1,\cdots,v-1,a_k=2,\gamma_k=1,\\
\tau^i_{k},&i=1,\cdots,v-1,\text{ otherwise},
\end{cases}
\end{align}
\begin{align}
\label{eta2}
\eta^i_{k+1} &= \begin{cases}
1,&i=0, a_k=2,\gamma_k=1,\\
\eta^{0}_{k}+1,& i=0, \text{ otherwise},\\
\eta^{i-1}_{k},&i=1,\cdots,v-1,a_k=2,\gamma_k=1,\\
\eta^i_{k},&i=1,\cdots,v-1,\text{ otherwise},
\end{cases}
\\ \label{phi2}
\phi^i_{k+1} &= \begin{cases}
\tau^0_{k}+1,&i=0, a_k=2,\gamma_k=1,\\
\phi^{i}_{k}+1,& i=0, \text{ otherwise},\\
\phi^{i-1}_{k},&i=1,\cdots,v-1,a_k=2,\gamma_k=1,\\
\phi^i_{k},&i=1,\cdots,v-1,\text{ otherwise.}
\end{cases}
\end{align}

Then, we can derive the plant-state covariance in a closed form in terms of the state parameters as follows.
\begin{proposition}\label{covariance_analysis}
	\normalfont
	The plant-state covariance $\mathbf{P}_{k}$ in time slot $k$~is
		\begin{equation} \label{one-state_covariance0}
	\begin{aligned}
	\mathbf{P}_{k} 
%	=  \sum_{i=1}^{\phi_k^0}\mathbf{A}^{i-1}\mathbf{R}(\mathbf{A}^\top)^{i-1} + (\mathbf{A}+\mathbf{B}\mathbf{K})^{\phi_k^0-\tau_k^1}\sum_{i=\tau_k^1+1}^{\phi_k^1}\mathbf{A}^{i-1}\mathbf{R}(\mathbf{A}^\top)^{i-1}\big((\mathbf{A}+\mathbf{B}\mathbf{K})^{\phi_k^0-\tau_k^1}\big)^\top+\cdots\\ &+(\mathbf{A}+\mathbf{B}\mathbf{K})^{\phi_k^0+\cdots+\phi_k^{v-2}-(\tau_k^1+\cdots+\tau_k^{v-1})}\sum_{i=\tau_k^{v-1}+1}^{\phi_k^{v-1}}\mathbf{A}^{i-1}\mathbf{R}(\mathbf{A}^\top)^{i-1}\big((\mathbf{A}+\mathbf{B}\mathbf{K})^{\phi_k^0+\cdots+\phi_k^{v-2}-(\tau_k^1+\cdots+\tau_k^{v-1})}\big)^\top\\
	&=
	\mathbf{F}(\phi^0_k)
	+ \sum_{i=0}^{v-2}
	\mathbf{G}\left(\sum_{j=0}^{i}\phi^j_k-\sum_{j=0}^{i}\tau^{j+1}_k,\right.\\
	&\left.\hspace{1.5cm}\mathbbm{1}(\phi^{i+1}_k>\tau^{i+1}_k)\left(\mathbf{F}(\phi^{i+1}_k)-\mathbf{F}(\tau^{i+1}_k)\right)\right),
	\end{aligned}
	\end{equation}
	%which is determined by the state parameters $\{\tau_k^1,\cdots,\tau_k^{v-1},\phi_k^0,\cdots,\phi_k^{v-1}\}$,
	where the summation operator has the property that
	$ \sum_{i=a}^{b} (\cdot)=0 $ if $a>b$, $\mathbf{F}(\cdot)$ is defined in \eqref{F}, and 
	\begin{equation}
	\mathbf{G}\left(x,\mathbf{Y}\right)\triangleq
	(\mathbf{ A+BK})^x\mathbf{ Y } ((\mathbf{ A+BK})^x)^\top.
	\end{equation}
\end{proposition}
\begin{proof}
	See Appendix A.
\end{proof}

\begin{remark}
	Proposition~\ref{covariance_analysis} states that the state covariance $\mathbf{P}_{k}$ of a $v$-step controllable plant is determined by $(2v-1)$ state parameters, i.e., $\tau^i_k$, $i=1,\cdots,v-1$ and $\phi^i_k$, $i=0,\cdots,v-1$. 
\end{remark}

\begin{remark} \label{remark:extension}
In practice, it is possible that the plant \eqref{system_model} is $\bar{v}$-step controllable, i.e., $(\mathbf{ A+BK})^{\bar{v}}=\mathbf{ 0}$, where $\bar{v}>v$; it is also possible that 
when the controller gain $\mathbf{K}$ is predetermined and fixed, one cannot find $\bar{v}\in \mathbb{N}$ such that $(\mathbf{A+BK})^{\bar{v}}=\mathbf{0}$.
Moreover, the plant may not be finite-step  controllable, i.e., one cannot find a set of $\mathbf{K}$ and $\bar{v}\in \mathbb{N}$ such that $(\mathbf{A+BK})^{\bar{v}}=\mathbf{ 0}$.
In these cases, where conditions \eqref{a+bk} and \eqref{a+bk2} are not satisfied, we can show that the covariance $\mathbf{P}_k$ has incountably infinite many values and cannot be expressed by finite number of state parameters as in Proposition~\ref{covariance_analysis}. Furthermore, the process $\{\mathbf{P}_k\}$ is  not stationary making the long-term average cost function \eqref{metric} difficult to evaluate. However, when $v$ is sufficiently large, $(\mathbf{A+BK})^{v}$ approaches $\mathbf{0}$ as $\rho(\mathbf{A+BK})<1$. Thus, the plant-state vector in \eqref{detail_representaiton} of the proof of Proposition~\ref{covariance_analysis} obtained by letting $(\mathbf{A+BK})^{v}=\mathbf{0}$, is still a good approximation of $\mathbf{x}_k$ for these cases, and hence Proposition~\ref{covariance_analysis} can be treated as a countable-state-space approximation of the plant-state covariance.
%As a consequence, the optimal scheduling policy for a 
\end{remark}

\subsection{Problem Formulation}
The uplink-downlink transmission-scheduling policy is defined as the sequence
$\left\lbrace a_1,a_2,\cdots,a_k,\cdots \right\rbrace$, where $a_k$ is the transmission-scheduling action in time slot $k$.
In the following, we optimize the transmission-scheduling  policy for both the one-step and multi-step controllable plants such that the average cost of the plant in \eqref{metric} is minimized\footnote{In this work, we only focus on the design of the scheduling policy $\{a_k\}$, when the controller gain $\mathbf{K}$ and the length of predictive control $v$ are given and fixed. In our future work, the controller gain, the length of predictive control and the scheduling sequence will be jointly optimized.},~i.e.,
\begin{equation} \label{problem1}
\min_{ a_1,a_2,\cdots,a_k,\cdots }\quad J =\lim_{K \to \infty}\frac{1}{K}\sum_{k=0}^{K-1}\text{Tr}\left(\mathbf{Q}\mathbf{P}_k\right).
\end{equation}

\section{One-Step Controllable Case: Optimal Transmission-Scheduling Policy}\label{sec:one}
We first investigate the optimal transmission scheduling policy for the one-step controllable case, as it will also shed some light onto the optimal policy design of general multi-step controllable cases.
Note that in this section and the following section, we focus on the static channel scenario, and the design method of the optimal scheduling policies can be extended to the fading channel scenario, which will be discussed in Section~\ref{sec:fading}.

\subsection{MDP Formulation}\label{sec:mdp}
From \eqref{P_covariance2}, \eqref{tau} and \eqref{phi},
the next state cost $\mathbf{P}_{k+1}$, and the next states $\tau_{k+1}$ and $\phi_{k+1}$ only depend on the current transmission-scheduling action $a_{k}$ and the current states $\tau_{k}$ and $\phi_{k}$. Therefore, we can reformulate the problem \eqref{problem1} into a Markov Decision Process (MDP) as follows.

1) The state space is defined as 
$\mathbb{S} \triangleq \{(\tau, \phi) : \phi \geq \tau, \phi \neq  \tau +1, \tau \in \mathbb{N}, \phi \in \{2,3,\cdots\}\}$ as illustrated in Fig.~\ref{fig:state_space}. Note that the states with $\phi = \tau +1$ are transient states (which can be verified using \eqref{tau} and \eqref{phi}) and are not included in $\mathbb{S}$, since we only focus on the long-term performance of the system.  
The state of the MDP at time $k$ is 
$\mathbf{s}_k \triangleq (\tau_k, \phi_k) \in \mathbb{S}$.  
\begin{figure}
	\centering
	\includegraphics[scale=0.7]{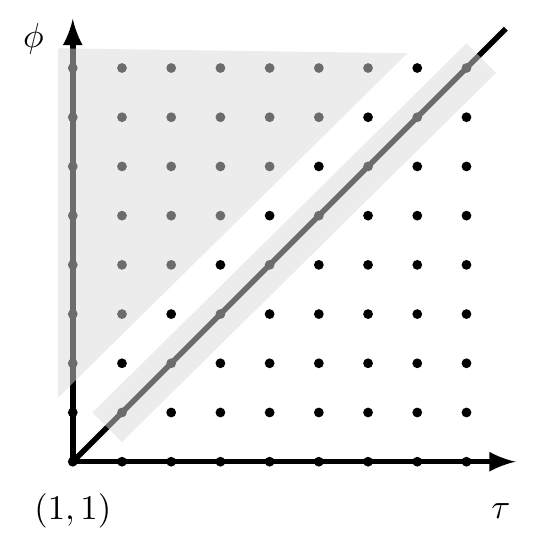}
\vspace{-0.5cm}	
	\caption{The state space $\mathbb{S}$ (shaded dots) of the MDP.}
	\label{fig:state_space}
	\vspace{-0.5cm}
\end{figure}

2) The action space of the MDP is defined as $\mathbb{A} \triangleq \{1,2\}$. The action at time $k$, $a_k \triangleq \pi(\mathbf{s}_k) \in \mathbb{A}$, indicates the sensor's transmission $(a_k=1)$ or the controller's transmission $(a_k=2)$ in time slot $k$.

3) The state-transition probability $P(\mathbf{s}'|\mathbf{s},a)$ is the probability that the state $\mathbf{s}$ at time $(k-1)$ transits to $\mathbf{s}'$ at time $k$ with action $a$ at time $(k-1)$. We drop the time index $k$ here since the transition is time-homogeneous. Let $\mathbf{s}=(\tau,\phi)$ and $\mathbf{s}'=(\tau',\phi')$ denote the current and next state, respectively. From \eqref{tau} and \eqref{phi}, the state-transition probability can be obtained as
\begin{equation} \label{transition_func}
P(\mathbf{s}'\vert \mathbf{s}, a)=
\left\lbrace
\begin{aligned}
& p_s, &&\text{if }a=1,\mathbf{s}'= (\tau+1,\phi+1) \\
& 1-p_s, &&\text{if }a=1,\mathbf{s}'= (1,\phi+1) \\
& p_c, &&\text{if }a=2,\mathbf{s}'= (\tau+1,\phi+1) \\
& 1-p_c, &&\text{if }a=2,\mathbf{s}'= (\tau+1,\tau+1) \\
& 0, &&\text{otherwise}.\\
\end{aligned}
\right.
\end{equation}
%where $p_s$ and $p_c$ are the packet-error probability of the uplink and downlink channels, respectively.

4) The one-stage cost of the MDP, i.e., the one-step quadratic-cost of the plant in \eqref{metric}, is a function of the current state $\phi$~as
\begin{equation} \label{one-stage cost}
c(\mathbf{s}) =c(\phi) \triangleq \text{Tr}\left(\mathbf{Q}\mathbf{P}\right) = \text{Tr}\left(\mathbf{Q}\mathbf{F}(\phi)\right),
\end{equation}
which is independent of the state $\tau$ and the action $a$. The function $c(\cdot)$ has the following property:
\begin{lemma}\label{lem:monotone}
	\normalfont
	The one-stage cost function $c(\phi)$ in \eqref{one-stage cost} is a strictly monotonically increasing function of $\phi$, where $\phi \in \{2,3,\cdots\}$.
\end{lemma}
\begin{proof}
Since $\mathbf{R}$ is a positive definite matrix, $\mathbf{MRM}^\top$ is positive definite for any $n$-by-$n$ non-zero matrix $\mathbf{M}$.
Also, we have $\mathbf{A}^i \neq \mathbf{0}, \forall i\in \mathbb{N}$, as it is assumed that $\rho(\mathbf{A})>1$ in Section~\ref{sec:plant}. Due to the fact that the product of positive-definite matrices has positive trace and $\mathbf{Q}$ is positive definite, $\text{Tr}\left(\mathbf{Q}\mathbf{A}^i \mathbf{R} (\mathbf{A}^i)^\top \right)$ is positive, $\forall i \in \mathbb{N}$.
From the definition of $\mathbf{F}(\cdot)$ in \eqref{F}, we have 
\begin{equation}
\begin{aligned}
c(\phi+z)-c(\phi) &=\text{Tr}(\mathbf{Q} \mathbf{F}(\phi+z))-\text{Tr}(\mathbf{Q}\mathbf{F}(\phi))\\
&=\sum_{i=\phi+1}^{\phi+z}
\text{Tr}\left(\mathbf{Q}\mathbf{A}^i \mathbf{R} (\mathbf{A}^i)^\top \right) > 0, \forall z\in\mathbb{N}.
\end{aligned}
\end{equation}
This completes the proof.
\end{proof}

Therefore, the problem \eqref{problem1} is equivalent to finding the optimal policy $\pi(\mathbf{s}s), \forall \mathbf{s}\in \mathbb{S}$ by solving the classical \emph{average cost minimization problem} of the MDP~\cite{puterman2014markov}. If a stationary and deterministic optimal policy of the MDP exists,
we can effectively find the optimal policy by using standard methods such as the relative value
iteration algorithm see e.g.,~\cite[Chapter 8]{puterman2014markov}.

\subsection{Existence of the Optimal Scheduling Policy}
If the uplink and downlink channels have high packet-error probabilities, the average cost in \eqref{problem1} may never be bounded no matter what policy we choose.
Therefore, \emph{we need to study the condition in terms of the transmission reliability of the uplink and downlink channels, under which the dynamic plant can be stabilized, i.e., the average cost can be bounded.}
We derive the following result.
%Since the cost function $c((\tau,\phi),a)$ grows exponentially with $\phi$, the average control cost could be unbounded under all the possible policy. We derive the necessary and sufficient condition that there exits a deterministic policy to minimize and stabilize the control cost.

\begin{theorem} \label{theorem:existence}
	\normalfont
In the static channel scenario, there exists a stationary and deterministic optimal transmission-scheduling policy that can stabilize the one-step controllable plant \eqref{system_model} iff 
\begin{equation}\label{condition}
\max \left\lbrace p_s, p_c\right\rbrace <\frac{1}{\rho^2(\mathbf{A})},
\end{equation}
where we recall that $\rho(\mathbf{A})$ is the spectral radius of $\mathbf{A}$.

%\begin{equation}\label{condition}
%	\rho^2(A) < \min \left\lbrace \frac{1}{p_s},\frac{1}{p_c} \right\rbrace.
%\end{equation}
\end{theorem}

\begin{proof}
The necessity of the condition can be easily proved as \eqref{condition} is the necessary and sufficient condition that an ideal FD controller with the uplink-downlink packet-error probabilities $\{p_s,p_c\}$ can stabilize the remote plant~\cite{schenato2007foundations}. Intuitively, if \eqref{condition} does not hold, an FD controller cannot stabilize the plant and thus an HD controller cannot either, no matter what transmission-scheduling policy it applies.

The sufficiency part of the proof is conducted by proving the existence of a stationary and deterministic policy $\pi'$ that can stabilize the plant if \eqref{condition} is satisfied, where
\begin{equation} \label{suboptimal}
\pi'(\mathbf{s}) = \pi'(\tau,\phi) =
\begin{cases}
1, \ \tau=\phi, (\tau,\phi)\in \mathbb{S}\\
2, \ \text{otherwise}.
\end{cases}
\end{equation}
The details of the proof are given in Appendix B.
\end{proof}

\begin{remark}
Theorem~\ref{theorem:existence} states that the optimal policy
exists, which stabilizes the plant, if both the channel conditions of the uplink and downlink channels are good (i.e., small $p_s$ and $p_c$) and the dynamic process does not change
rapidly (i.e., a small $\rho^2(\mathbf{A})$).
Also, it is interesting to see that the HD controller has exactly the same condition with the FD controller~\cite{schenato2007foundations} to stabilize the plant. However, since the HD operation naturally introduces longer delay in both transmissions of the sensor measurement and the control command than the FD operation, the bounded average cost of the HD controller should be higher than that of the FD one, which will be illustrated in Section~\ref{sec:num}.
\end{remark}

%As it is hard to give the close form expression, we derive structure property of the optimal schedule policy, which can help to reduce computational complexity and save storage space for online implementation.

Assuming that the condition~\eqref{condition} is satisfied, we have the following property of the optimal policy.
\begin{proposition} \label{theorem:switching}
   	\normalfont		
   	The stationary and deterministic optimal policy of the problem \eqref{problem1}, $\pi^*(\tau,\phi)$, is a switching-type policy in terms of $\tau$ and $\phi$, i.e., 
   	(i) if $\pi^{*}(\tau,\phi)=1$, then  $\pi^{*}(\tau+z,\phi)=1$, $\forall z\in\mathbb{N}$ and $(\tau +z,\phi)\in \mathbb{S}$;
   	(ii) if $\pi^{*}(\tau,\phi)=2$, then  $\pi^{*}(\tau,\phi+z)=2$, $\forall z\in\mathbb{N}$ and $(\tau,\phi  +z)\in \mathbb{S}$.
\end{proposition}
\begin{proof}
The proof follows the same procedure as that of \cite[Theorem 2]{huang2019retransmit} and is omitted due to the space limitation.
\end{proof}
Therefore, for the optimal policy, the state space is divided into two parts by a curve, and the scheduling actions of the states in each part are the same, which will be illustrated in Section~\ref{sec:num}.
Such a switching structure helps saving
storage space for on-line transmission scheduling, as the controller only needs to store the states of the switching boundary instead of the entire state space~\cite{huang2019retransmit,TWC}. 
%At each time, the
%sensor simply needs to compare the current state with the
%boundary states to give the optimal decision.

\subsection{Suboptimal Policy} \label{sec:sub}
In practice, to solve the MDP problem in Section~\ref{sec:mdp} with an infinite number of states, one needs to approximate it by a truncated MDP problem
with finite states for offline numerical evaluation. The computing complexity of the problem is $\mathcal{O}(A B^2 C)$~\cite{littman1995complexity}, where $A$ and $B$ are the cardinalities of the action space and the state space, respectively, and $C$ is the number of convergence steps for solving the problem.
To reduce the computation complexity,
we propose a myopic policy $\psi(\mathbf{s}), \forall \mathbf{s}\in \mathbb{S}$, which simply makes online decision to optimize the expected next stage cost.

From \eqref{transition_func} and \eqref{one-stage cost}, the expected next stage cost $\mathbb{E}[c(\phi')\vert \mathbf{s}, a=\psi(\mathbf{s})]$, where $\mathbf{s}=(\tau,\phi)$, is derived as
\begin{equation} \label{compare}
\begin{aligned}
&   \mathbb{E}[c(\phi')\vert \mathbf{s}, \psi(\mathbf{s})=1] = c(\phi+1),\\
&   \mathbb{E}[c(\phi')\vert \mathbf{s}, \psi(\mathbf{s})=2] = p_c c(\phi+1) 
+ (1-p_c) c(\tau+1).
\end{aligned}
\end{equation} 

1) For the states $\{\mathbf{s} \vert (\tau,\phi)\in\mathbb{S}, \phi>\tau \}$, from \eqref{compare}, the action \emph{$\psi(\mathbf{s})=2$ results in a smaller next stage cost than $\psi(\mathbf{s})=1$}.

2) For the states $\{\mathbf{s}\vert (\tau,\phi)\in\mathbb{S}, \phi=\tau \}$, from \eqref{compare}, since the two actions lead to the same next stage cost, i.e., \begin{equation} \label{sub}
\begin{aligned}
\mathbb{E}[c(\phi')\vert \mathbf{s}, \psi(\mathbf{s})=1] &=\mathbb{E}[c(\phi')\vert \mathbf{s}, \psi(\mathbf{s})=2] = c(\phi+1),
\end{aligned}
\end{equation}
we need to compare the second stage cost led by the actions. If $\psi(\mathbf{s})=1$, $\mathbf{s}'\in \{(1,\phi+1),(\phi+1,\phi+1)\}$. If $\mathbf{s}'=(1,\phi+1)$, since $\phi+1>1$, the next stage myopic action is $\psi(1,\phi+1)=2$ as discussed earlier and the second stage state $\mathbf{s}''\in \{(2,\phi+2),(2,2)\}$. If $\mathbf{s}'=(\phi+1,\phi+1)$, from \eqref{sub}, the expected second stage cost is $c(\phi+2)$ for both $\psi(\mathbf{s}')=1 $ and~$2$.
Based on these analysis and \eqref{transition_func}, we have the expected second stage cost with $\phi(\mathbf{s})=1$ as 
\begin{equation}
\begin{aligned}
   \mathbb{E}[c(\phi'')\vert \mathbf{s}, \psi(\mathbf{s})=1] &= 
(1-p_s)\left(p_c c(\phi+2) \!+\! (1-p_c) c(2)\right)\\
&+p_s c(\phi+2).
\end{aligned}
\end{equation} 
Similarly, we can obtain the expected second stage cost with $\phi(\mathbf{s})=2$~as 
\begin{equation}
\begin{aligned}
&   \mathbb{E}[c(\phi'')\vert \mathbf{s}, \psi(\mathbf{s})=2] = c(\phi+2).
\end{aligned}
\end{equation} 
Since $p_c,p_s<1$ and $c(2)<c(\phi+2)$ from Lemma~\ref{lem:monotone}, $\psi(\mathbf{s})=1$ \emph{results in a smaller cost than} $\psi(\mathbf{s})=2$. 
From the above analysis, the myopic policy $\psi(\mathbf{s})$ is equal to $\pi'(\mathbf{s})$ in \eqref{suboptimal}, $\forall \mathbf{s}\in \mathbb{S}$.
\begin{proposition} \label{prop:sub}
	\normalfont
The myopic policy of problem \eqref{problem1} is $\pi'$ in~\eqref{suboptimal}.
\end{proposition}
\begin{remark}
From the myopic policy \eqref{suboptimal} and the state-updating rules \eqref{tau} and \eqref{phi}, we see that the policy $\pi'$ is actually a {persistent scheduling policy}, which consecutively schedules the uplink transmission until a transmission is successful and then consecutively schedules the downlink transmission until a transmission is successful, and so on.
\end{remark}
From the property of the persistent scheduling policy, we can easily obtain the result below.
\begin{corollary}
	For the persistent uplink-downlink scheduling policy $\pi'$ in~Proposition~\ref{prop:sub}, the chances for scheduling the sensor's and the controller's transmissions, are $\frac{1-p_c}{(1-p_c)+(1-p_s)}$ and $\frac{1-p_s}{(1-p_c)+(1-p_s)}$, respectively.
\end{corollary}

\subsection{Naive Policy: A Benchmark} \label{sec:naive}
We consider a naive uplink-downlink scheduling policy of the HD controller, as a benchmark of the proposed optimal scheduling policy.
The naive policy simply schedules the sensor's and the controller's transmissions alternatively, i.e., \{$\cdots$, sensing, control, sensing, control, $\cdots$\}, without taking into account the state-estimation quality of the controller nor the state-quality of the plant. Such a naive policy is also noted as the \emph{round-robin scheduling policy}.

\begin{theorem} \label{theorem:naive}
	\normalfont
In the static channel scenario, the alternative scheduling policy can stabilize the one-step controllable plant \eqref{system_model} iff 
	\begin{equation}\label{condition2}
\max \left\lbrace p_s, p_c\right\rbrace <\frac{1}{\left(\rho^2(\mathbf{A})\right)^2}.
	\end{equation}
\end{theorem}
%\begin{theorem} \label{theorem:naive}
%	\normalfont
%	The alternating-schedule policy can stabilize the plant~\eqref{system_model} iff 
%	\begin{equation}\label{condition2}
%	\left(\rho^2(A)\right)^2 < \min \left\lbrace \frac{1}{p_s},\frac{1}{p_c} \right\rbrace.
%	\end{equation}
%\end{theorem}
\begin{proof}
See Appendix C.
\end{proof}
\begin{remark}
Comparing with Theorem~\ref{theorem:existence}, to stabilize the same plant, the naive policy may require  smaller packet-error probabilities of the uplink and downlink channels than the proposed optimal scheduling policy. This also implies that the optimal policy can result in a notably smaller the average cost of the plant than the naive policy, which will be illustrated in Section~\ref{sec:num}.
\end{remark}

\section{$v$-Step Controllable Case: Optimal Transmission-Scheduling Policy} \label{sec:v}
%\subsection{Estimation-Error Covariance}
%The estimation process for $v$-step controllable plant is the same as $1$-step controllable plant, i.e.
%\begin{equation} \label{error_representation}
%\mathbf{e}_k \triangleq \mathbf{x}_{k} - \hat{\mathbf{x}}_{k} = \sum_{i=1}^{\tau_k}\mathbf{A}^{i-1}\mathbf{w}_{k-i}
%\end{equation}
%\begin{equation}
%\mathbb{U}_{k} \triangleq\mathbb{E}[\mathbf{e}_{k}\mathbf{e}_{k}^\top] = \sum_{i=1}^{\tau_k}\mathbf{A}^{i-1}\mathbf{R}(\mathbf{A}^\top)^{i-1}
%\end{equation}

%\subsection{Plant-State Covariance}

\subsection{MDP Formulation}\label{sec:mdp2}
Based on Proposition~\ref{covariance_analysis}, the average cost minimization problem \eqref{problem1} can be formulated as an MDP similar to 
the one-step controllable case in Section~\ref{sec:one} as:

1) The state space is defined as 
$\mathbb{S} \triangleq \{(\tau_k^0,\tau_k^1,\cdots,\tau_k^{v-1},\phi_k^0,\phi_k^1,\cdots,\phi_k^{v-1}):\phi^i_k \geq \tau^i_k, \phi^i_k \neq  \tau^i_k +1, \tau^i_k \in \mathbb{N}, \phi^i_k \in \{2,3,\cdots\},\forall i=0,\cdots,v-1\}$.

2) The action space of the MDP is exactly the same as that of the one-step controllable plant in Section~\ref{sec:mdp}.

3) Let $P(\mathbf{s}'|\mathbf{s},a)$ denote the state-transition probability, where $\mathbf{s} = (\tau^0,\cdots,\tau^{v-1},\phi^0,\cdots,\phi^{v-1})$ and $\mathbf{s}' = ((\tau^0)',\cdots,(\tau^{v-1})',(\phi^0)',\cdots,(\phi^{v-1})')$ are the current and next state, respectively, after dropping the time indexes. 
From \eqref{tau2} and \eqref{phi2}, the state-transition probability is obtained as
\begin{equation}
\begin{aligned}
&P(\mathbf{s}'\vert \mathbf{s}, a)=\\
&
\left\lbrace
\begin{aligned}
& p_s, \text{if }a\!=\!1,\mathbf{s}'\!=\! (\tau^0\!+\!1,\tau^1,\cdots,\tau^{v-1},\phi^0\!+\!1,\phi^1,\cdots,\phi^{v\!-\!1}) \\
& 1\!-\!p_s, \text{if }a\!=\!1,\mathbf{s}'\!=\! (1,\tau^1,\cdots,\tau^{v-1},\phi^0\!+\!1,\phi^1,\cdots,\phi^{v-1}) \\
& p_c, \text{if }a\!=\!2,\mathbf{s}'\!=\! (\tau^0\!+\!1,\tau^1,\cdots,\tau^{v-1},\phi^0\!+\!1,\phi^1,\cdots,\phi^{v-1}) \\
& 1\!-\!p_c, \text{if }a\!=\!2,\mathbf{s}'\!=\! (\tau^0\!+\!1,\tau^0,\!\cdots,\!\tau^{v\!-\!2},\tau^0\!+\!1,\phi^0,\!\cdots,\!\phi^{v\!-\!2}) \\
& 0, \text{otherwise}.\\
\end{aligned}
\right.
\end{aligned}
\end{equation}
%The next state is updated as
%\begin{equation}
%\eta'^+ = \begin{cases}
%1 &\mbox{ if } a = 2, \gamma_k = 1\\
%\eta'+1 &\mbox{ otherwise } 
%\end{cases}
%\end{equation}
%\begin{equation}
%\tau^+ = \begin{cases}
%1 &\mbox{ if } a = 1, \delta_k = 1\\
%\tau+1 &\mbox{ otherwise } 
%\end{cases}
%\end{equation}
%for $i = \{1,\cdots,v-1\}$,
%\begin{equation}
%\eta_i^+ = \begin{cases}
%\eta_{i-1} &\mbox{ if } a = 2, \gamma_k = 1\\
%\eta_i &\mbox{ otherwise } 
%\end{cases}
%\end{equation}
%for $i = \{1,\cdots,v\}$,
%\begin{equation}
%\tau_i^+ = \begin{cases}
%\tau_{i-1} &\mbox{ if } a = 2, \gamma_k = 1\\
%\tau_i &\mbox{ otherwise } 
%\end{cases}
%\end{equation}
%where $\eta_0$ and $\tau_0$ represent $\eta'$ and $\tau$, respectively.

4) The one-stage cost of the MDP is a function of the current state $\mathbf{s}$, and is obtained from \eqref{metric} and Proposition~\ref{covariance_analysis} as
\begin{equation} \label{cost2}
\begin{aligned}
&c(\mathbf{s})=c(\tau^1,\cdots,\tau^{v-1},\phi^0,\cdots,\phi^{v-1})\\
%&= \text{Tr}\left(\mathbf{Q} \left[\mathbf{F}(\phi^0)
%+ \sum_{i=0}^{v-2} (\mathbf{A+BK})^{\sum_{j=0}^{i}\phi^j-\sum_{j=0}^{i}\tau^{j+1}}
%\left(\mathbf{F}(\phi^{i+1})-\mathbf{F}(\tau^{i+1})\right)
% \left((\mathbf{A+BK})^{\sum_{j=0}^{i}\phi^j-\sum_{j=0}^{i}\tau^{j+1}}\right)^\top\right]\right)\\
&= \text{Tr}\left(\mathbf{Q} \left[\mathbf{F}(\phi^0)
+ \sum_{i=0}^{v-2}
\mathbf{G}\left(\sum_{j=0}^{i}\phi^j-\sum_{j=0}^{i}\tau^{j+1},\right.\right.\right.\\
&\left.\left.\left.\hspace{1.8cm}\mathbbm{1}(\phi^{i+1}>\tau^{i+1})\left(\mathbf{F}(\phi^{i+1})-\mathbf{F}(\tau^{i+1})\right)\right)\right]
\right).
\end{aligned}
\end{equation}
\begin{remark}
	Different from the one-step controllable case, where the one-stage cost function is a monotonically increasing function of the state parameter $\phi$, the cost function in \eqref{cost2} is more complex and does not have such a property. Thus, the switching structure of the optimal policy does not hold in general for the $v$-step controllable case.
\end{remark}

\subsection{Existence of the Optimal Scheduling Policy}\label{sec:v-result}

\begin{theorem} \label{theorem:existence2}
	\normalfont
	In the static channel scenario, there exists a stationary and deterministic optimal transmission-scheduling policy that can stabilize the $v$-step controllable plant \eqref{system_model} using the predictive control method \eqref{C}, \eqref{U} and \eqref{u}, iff \eqref{condition} holds.
\end{theorem}
\begin{proof}
	See Appendix D.
\end{proof}
\begin{remark}
	The stability condition of a $v$-step controllable plant is exactly the same as that of the one-step controllable plant in Theorem~\ref{theorem:existence}. Thus, whether a plant can be stabilized by an HD controller simply depends on the spectral radius of the plant parameter $\mathbf{A}$ and the uplink and downlink transmission reliabilities.
\end{remark}

\begin{remark}
Although the stability conditions of a one-step and a $v$-step plants are the same, 
to find the optimal uplink-downlink scheduling policy, the state space and the computation complexity of the MDP problem grow up with $v$ linearly and exponentially~\cite{littman1995complexity}, respectively.
However, in the following section, we will show that the persistent scheduling policy in Proposition~\ref{prop:sub}, which can be treated as a policy that makes decision simply relying on two state parameters, i.e., $\phi^0$ and $\tau^0$, instead of the entire $2 v$ state parameters, can provide a remarkable performance close to the optimal one.
\end{remark}

\section{Extension to Fading Channels} \label{sec:fading}
{In this section, we investigate the optimal transmission-scheduling policy for the general $v$-step controllable case in the fading channel scenario, where $v\geq 1$. 
\subsection{MDP Formulation}\label{sec:fading1}
Comparing with the static channel scenario, the transmission scheduling of the WNCS in the fading channel scenario should take into account the channel states of both the uplink and downlink channels, and hence expands the dimension of the state space. Also, the state-transition probabilities of the MDP problem should also rely on the transition probabilities of channel states.
Therefore, the detailed MDP problem for solving the average cost minimization problem~\eqref{problem1} can be formulated as:
%
%
%
%. And the state-transition probability will also be effected by the Markov channel state transition probability. The detailed MDP formulation is given as:

1)  The state space is defined as 
$\mathbb{S} \triangleq \{(\tau_k^0,\tau_k^1,\cdots,\tau_k^{v-1},\phi_k^0,\phi_k^1,\cdots,\phi_k^{v-1},h_{s,k},h_{c,k}):\phi^i_k \geq \tau^i_k, \phi^i_k \neq  \tau^i_k +1, \tau^i_k \in \mathbb{N}, \phi^i_k \in \{2,3,\cdots\},h_{s,k}\in\{1,\cdots,{B_s}\},h_{c,k}\in\{1,\cdots,{B_c}\},\forall i=0,\cdots,v-1 \}$, where $h_{s,k}$ and $h_{c,k}$ are channel-state indexes of the uplink and downlink channels at time $k$, respectively.

2) The action space is the same as that of the static channel scenario in Section~\ref{sec:mdp2}.

3) As the state transition is time-homogeneous, we drop the time index $k$ here. Let $h \triangleq (h_s,h_c)$ and $h' \triangleq (h_s',h_c')$ denote the current and the next uplink-downlink channel states, respectively. As the uplink and downlink channel are action-invariant and independent of each other, the overall channel state transition probability can be directly obtained from \eqref{channel_transition1} and \eqref{channel_transition2} as 
\begin{equation}
P(h'\vert h) = d^s_{h_s,h_s'} d^c_{h_c,h_c'}.
\end{equation}
Let $\mathbf{s} \triangleq (\tau^0,\cdots,\tau^{v-1},\phi^0,\cdots,\phi^{v-1},h)$ and $\mathbf{s}' \triangleq ((\tau^0)',\cdots,(\tau^{v-1})',(\phi^0)',\cdots,(\phi^{v-1})',h')$ denote the current and the next states of the WNCS, respectively.
The state-transition probability $P(s'|s,a)$ can be obtained as
\begin{equation}
\begin{aligned}
&P(\mathbf{s}'\vert \mathbf{s}, a)=\\
&
\left\lbrace
\begin{aligned}
& P(h'\vert h) \omega_{h_s}, \text{if }a=1 \text{ and}\\
&\hspace{1.5cm}\mathbf{s}'= (\tau^0+1,\cdots,\tau^{v-1},\phi^0+1,\cdots,\phi^{v-1},h'), \\
& P(h'\vert h)(1-\omega_{h_s}), \text{if }a=1\text{ and}\\
&\hspace{1.5cm}\mathbf{s}'= (1,\cdots,\tau^{v-1},\phi^0+1,\cdots,\phi^{v-1},h'), \\
& P(h'\vert h)\xi_{h_c}, \text{if }a=2\text{ and}\\
&\hspace{1.5cm}\mathbf{s}'= (\tau^0+1,\cdots,\tau^{v-1},\phi^0+1,\cdots,\phi^{v-1},h'), \\
& P(h'\vert h)(1-\xi_{h_c}),\text{if }a=2\text{ and}\\ 
&\hspace{1.5cm}\mathbf{s}'= (\tau^0+1,\cdots,\tau^{v-2},\tau^0+1,\cdots,\phi^{v-2},h'), \\
& 0, \text{ otherwise}.\\
\end{aligned}
\right.
\end{aligned}
\end{equation}

4) The one-stage cost of the MDP is the same as \eqref{cost2}.

Such an MDP problem with $(2v+2)$ state dimensions and a small action space can be solved by standard MDP algorithms similar to that of the static channel scenario discussed earlier.

\subsection{Existence of the Optimal Scheduling Policy}
In the fading channel scenario, since each state of the Markov chain induced by a scheduling policy has $(2 v+2)$ dimensions, it is difficult to analyze the average cost of the Markov chain and determine whether it is bounded or not. Therefore, it is hard to give a necessary and sufficient condition in terms of the properties of the Markov channels and the plant, under which the MDP problem has a scheduling-policy solution leading to a bounded minimum average cost.
However, inspired by the result of the static channel scenario in Section~\ref{sec:v-result}, we can directly give a necessary condition and a sufficient condition by considering the best and the worst Markov channel conditions of the uplink and downlink channels as below.
\begin{theorem}\label{theorem:fading}
	\normalfont
	In the fading channel scenario, a necessary condition and a sufficient condition on the exists a stationary and deterministic optimal transmission-scheduling policy that can stabilize the general $v$-step controllable plant \eqref{system_model} using the predictive control method \eqref{C}, \eqref{U} and \eqref{u} are given by
\begin{equation}\label{nec}
\max \left\lbrace \underline{p_s}, \underline{p_c}\right\rbrace <\frac{1}{\rho^2(\mathbf{A})},
\end{equation}
and
\begin{equation}\label{sufficient}
\max \left\lbrace \overline{p_s}, \overline{p_c}\right\rbrace <\frac{1}{\rho^2(\mathbf{A})},
\end{equation}
respectively, where $\underline{p_s}\triangleq \min\{\omega_1,\cdots,\omega_{B_s}\}$, $\overline{p_s}\triangleq \max\{\omega_1,\cdots,\omega_{B_s}\}$, $\underline{p_c}\triangleq \min\{\xi_1,\cdots,\xi_{B_c}\}$, $\overline{p_c}\triangleq \max\{\xi_1,\cdots,\xi_{B_c}\}$.
\end{theorem}

In general, Theorem~\ref{theorem:fading} says that the plant can be stabilized by a transmission scheduling policy as long as the worst achievable channel conditions of the uplink and downlink Markov channels are good enough, and it cannot be stabilized by any scheduling policy if the best achievable channel conditions of the uplink and downlink Markov channels are poor enough.

In the following section, we will numerically evaluate the performance of the plant using the optimal transmission scheduling policy, where the sufficient condition of the existence of an optimal policy~\eqref{sufficient} is satisfied.}

\begin{figure}
	\centering
	\includegraphics[scale=0.27]{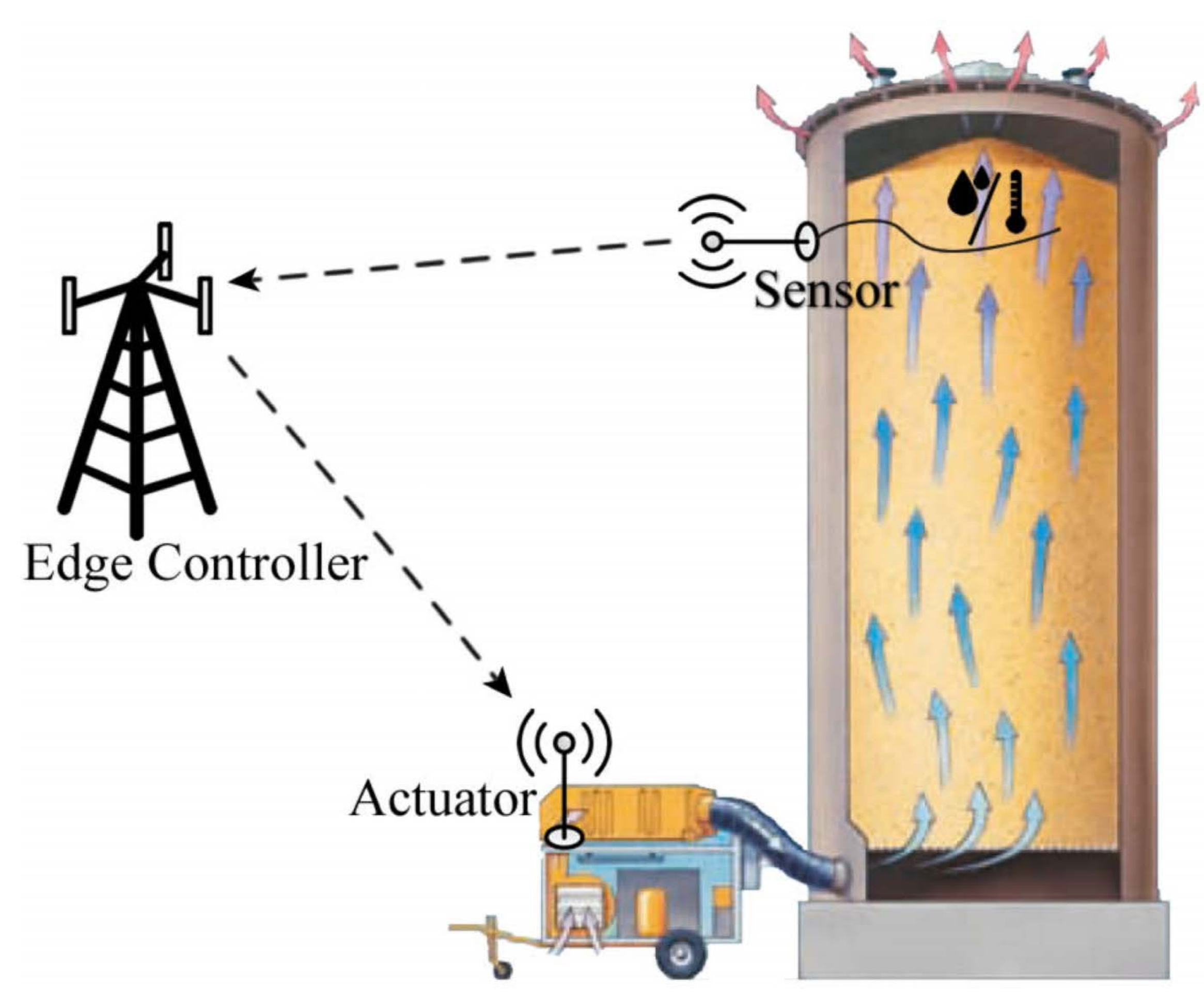}
	\vspace{-0.2cm}
	\caption{Temperature and humidity control in grain conservation.}
	\label{fig:farm}
	\vspace{-0.5cm}
\end{figure}
\section{Numerical Results} \label{sec:num}
{The uplink-downlink scheduling policies that we developed can be applied to a large range of real IIoT applications, including temperature control of hot rolling process in the iron and steel industry, flight path control of delivery drones, voltage control in smart grids, and lighting control in smart homes/buildings.
Specifically, in this section, we apply the uplink-downlink scheduling policies to a real application of smart farms as illustrated in Fig.~\ref{fig:farm}.
The system contains a grain container, a sensor measuring the temperature ($^{\circ}$C) and the humidity ($\%$) of the grain pile, an actuator which has a high pressure fan and/or an air cooler, and an edge controller, which receives the sensor's measurements, and then computes and sends the command to the actuator.
Given the present values of temperature and humidity, the state vector $\mathbf{x}_k$ in \eqref{system_model} contains two parameters, i.e., the current temperature and humidity offsets.
Note that since grain absorbs water from the air and generates heat naturally, the temperature and the humidity levels of the grain pile will automatically increase without proper control, leading to severe insect and mold development~\cite{grain}.
In general, by using the high pressure fan for ventilation, both the temperature and the humidity can be controlled in a proper range. If the air cooler is also available, the temperature can be controlled to the preset value faster.
Thus, if the actuator has a high pressure fan only, given a preset fan speed, its control input $\mathbf{u}_k$ in \eqref{system_model} has only one parameter that is the relative fan speed (measured by the flow volume [m$^3$/h])
If the actuator has both high pressure fan and air cooler, the control input $\mathbf{u}_k$ has two parameters including the fan speed and the cooler temperature ($^{\circ}$C). The former and the latter cases will be studied in Section~\ref{double}, and Sections~\ref{single1} and \ref{single2}, respectively.}

{The discrete time step $T_0$ in this example is set to be one second~\cite{GatsisOppor}.
Unless otherwise stated, we assume the system parameters as
$\mathbf{A} = \begin{bmatrix}
1.1 & 0.2 \\
0.2 & 0.8
\end{bmatrix}$, $\mathbf{Q} = \mathbf{R} = \begin{bmatrix}
1 & 0\\
0 &1
\end{bmatrix}$, 
and thus
$\rho^2(\mathbf{A}) = 1.44$.
Since the controller, the sensor and the actuator have very low mobility in this example, we focus on the static channel scenario and set the packet error probabilities of the uplink and downlink channels as $p_s=0.1$ and $p_c=0.1$, respectively, and also study the fading channel scenario in Section~\ref{double}.}
%From Theorems~\ref{theorem:existence} and ~\ref{theorem:existence2}, the optimal policies of the two cases exists iff $p_s,p_c \leq 1/\rho^2(\mathbf{A})=0.7$.

{In the following, we present numerical results of the optimal policies and the optimal average costs of the plant in Sections~\ref{sec:one} and \ref{sec:v} for one-step and $v$-step controllable cases, respectively. Also, we numerically compare the performance of the optimal scheduling policy with the persistent scheduling policy in Section~\ref{sec:sub}, the benchmark (naive) policy in Section~\ref{sec:naive}, and also the ideal FD policy in~\cite{schenato2007foundations}, i.e., the controller works in the FD mode and have the same packet-error probabilities of the uplink-downlink channels as in the HD mode.}

%When calculating the MDP problem in Sections~\ref{sec:mdp} and ~\ref{sec:mdp2}, 
{Note that to calculate the optimal policies in Sections~\ref{sec:mdp}, \ref{sec:mdp2} and~\ref{sec:fading1} by solving the MDP problems,
the infinite state space $\mathbb{S}$ is first truncated by limiting the range of the state parameters as $1 \leq \tau^i,\phi^i \leq 20, \forall i=0,\cdots, v-1$, to enable the evaluation. 
For example, if we consider a two-step controllable case, i.e., $v=2$, there will be $20^{2\times v}=160,000$ states in the static channel scenario, and there will be much more states in the fading channel scenario.
For solving finite-state MDP problems, in general, there are two classical methods: the policy iteration and the value iteration.
The policy iteration method converges faster in solving small-scale MDP problems, but is more computationally burdensome than the value iteration method when the state space is large~\cite{puterman2014markov}. Since our problems have large state spaces, we adopt the classical value iteration method for solving the MDP problems by using a well recognized Matlab MDP toolbox~\cite{matlab}.}

\subsection{One-Step Controllable Case}\label{double}
In this case, we assume that $\mathbf{B} = - \begin{bmatrix}
1 & 0 \\
0 & 1
\end{bmatrix}$, and $\mathbf{K}=\mathbf{A}$ satisfying $\mathbf{A+BK}=\mathbf{ 0}$.

\emph{\underline{Optimal and suboptimal policies.}}
Fig.~\ref{fig:policy} shows the optimal policy and the persistent (suboptimal) policy in Proposition~\ref{prop:sub} within the truncated state space. We see that although the optimal policy has more states choosing to schedule the sensor's transmission than the persistent policy, these two policies look similar to each other. Also, we see that the optimal policy is a switching-type policy in line with Proposition~\ref{theorem:switching}.

\begin{figure}
	\centering
	\includegraphics[scale=0.67]{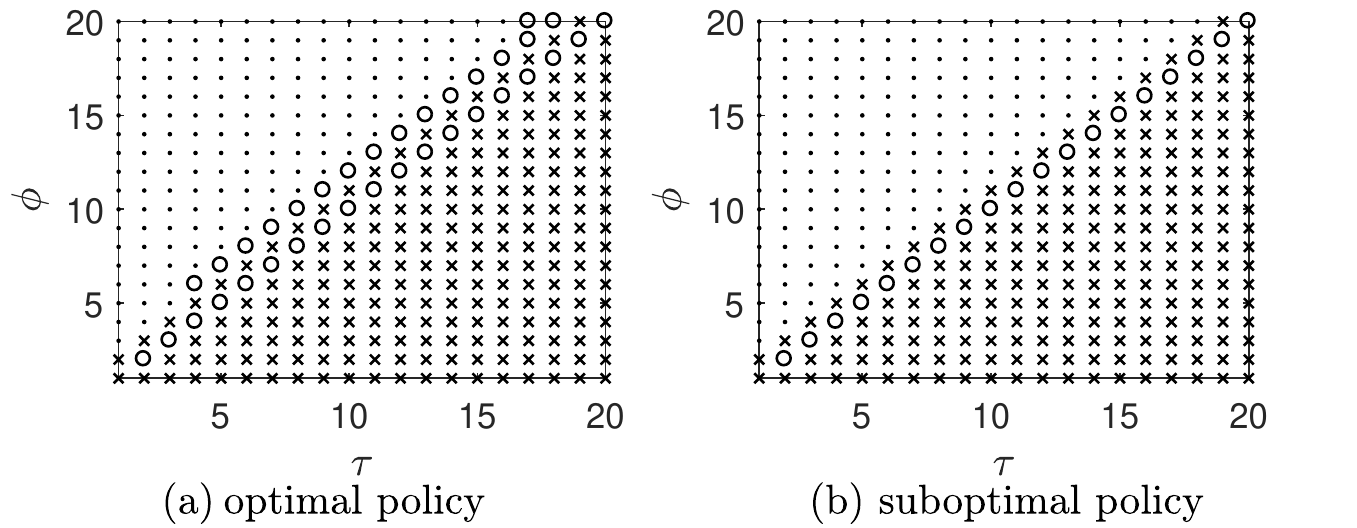}
	\vspace{-0.5cm}
	\caption{The uplink-downlink scheduling policies, where where `o' and
		`.' denote $a = 1$ and $a = 2$, respectively, and `x' denotes a state that does not belong to~$\mathbb{S}$.}
	\label{fig:policy}
	\vspace{-0.5cm}
\end{figure}

\emph{\underline{Performance comparison.}} 
We further evaluate the performances of the optimal scheduling policy, the persistent policy, the naive policy and also the FD policy
in terms of the $K$-step average cost of the plant using $\frac{1}{K}\sum_{k=0}^{K-1}  \mathbf{x}_k^\top\mathbf{Q}\mathbf{x}_k$.
%$\frac{1}{K}\sum_{k=1}^{K} trP_k.$
%, and the delay, i.e.,
%$\frac{1}{K}\sum_{k=1}^{K} \tau_k,$
%where $K$ is the simulation time.
We run $10^4$-step simulations with the initial value of the plant-state vector $\mathbf{x}_0=[1,-1]^\top$. The initial state for the optimal and persistent policies is $(\tau_0,\phi_0) =(2,2)$. The initial scheduling of the naive policy is the sensor's transmission.
%the estimation-error covariance, $\mathbf{U}_0$, and the plant-state covariance, $\mathbf{P}_0$, as $\mathbf{U}_0 = \mathbf{P}_0 = \mathbf{F}(2) = \begin{bmatrix}
%2.25  & 0.38\\
%0.38  &  0.68
%\end{bmatrix}$.
%and the initial value of $\tau_k$: $\tau_0 = 1$.
%From Lemma~\ref{lem:monotone}, we set the minimum one-stage cost $c(\phi=2) =~3.93$ as a \emph{performance baseline}, as the long-term average cost of any HD policy should be larger than $c(2)$.

Fig.~\ref{fig:performance} shows the average cost versus the simulation time, using different policies.
We see that the average costs induced by different policies converge to the steady state values when $K > 3000$.
%Also, we see that the performance of the persistent policy is very close to the optimal one.
Given the baseline of the FD (non-scheduling) policy, the optimal scheduling policy gives a significant $60\%$ average cost reduction of the naive policy.
Also, we see that the persistent policy provides a performance close to the optimal one.
We note that there is a noticeable performance gap between the optimal scheduling policy of the HD controller and the FD policy of the FD controller, since the HD operation introduces extra delays in uplink-downlink transmissions and deteriorates the performance of the control system.

%\begin{figure*}
%	\begin{minipage}{0.5\textwidth}
%	\centering
%\includegraphics[scale=0.58]{simulation.pdf}
%\vspace{-0.5cm}	
%\caption{Average cost versus time.}
%\label{fig:performance}
%\vspace{-0.6cm}
%	\end{minipage}
%%
%	\begin{minipage}{0.5\textwidth}
%	\centering
%%	\vspace{0.6cm}
%\includegraphics[scale=0.58]{average_cost.pdf}
%\vspace{-0.5cm}	
%\caption{Average cost versus packet-error probabilities, i.e., $p_s$ and $p_c$.}
%\label{fig:average_cost}
%\vspace{-0.6cm}
%\end{minipage}
%\end{figure*}

\begin{figure}[t]
	\centering
	\includegraphics[scale=0.58]{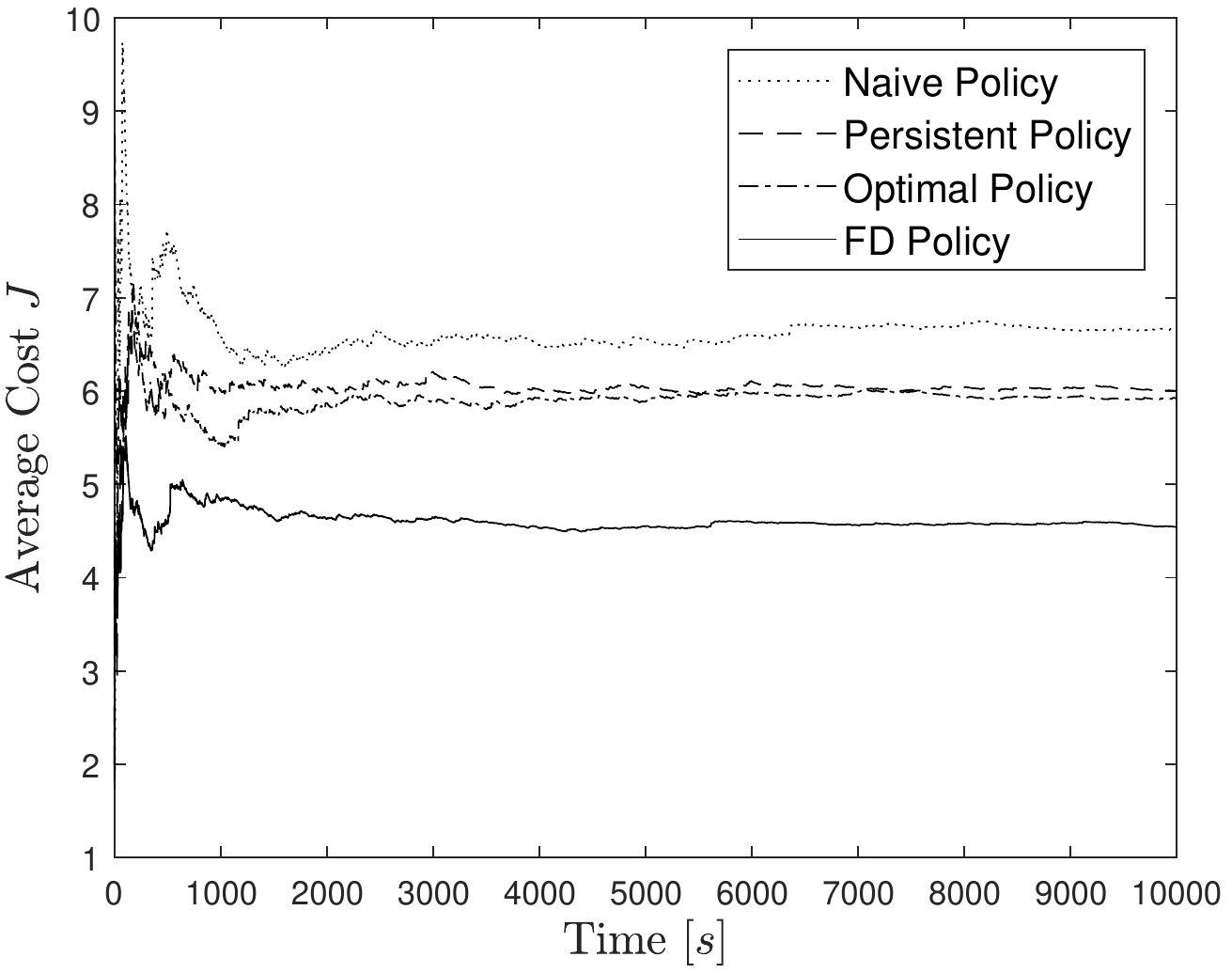}
\vspace{-0.3cm}	
	\caption{One-step controllable case: average cost versus time.}
	\label{fig:performance}
	\vspace{-0.0cm}
\end{figure}

\begin{figure}[t]
	\centering
	\includegraphics[scale=0.58]{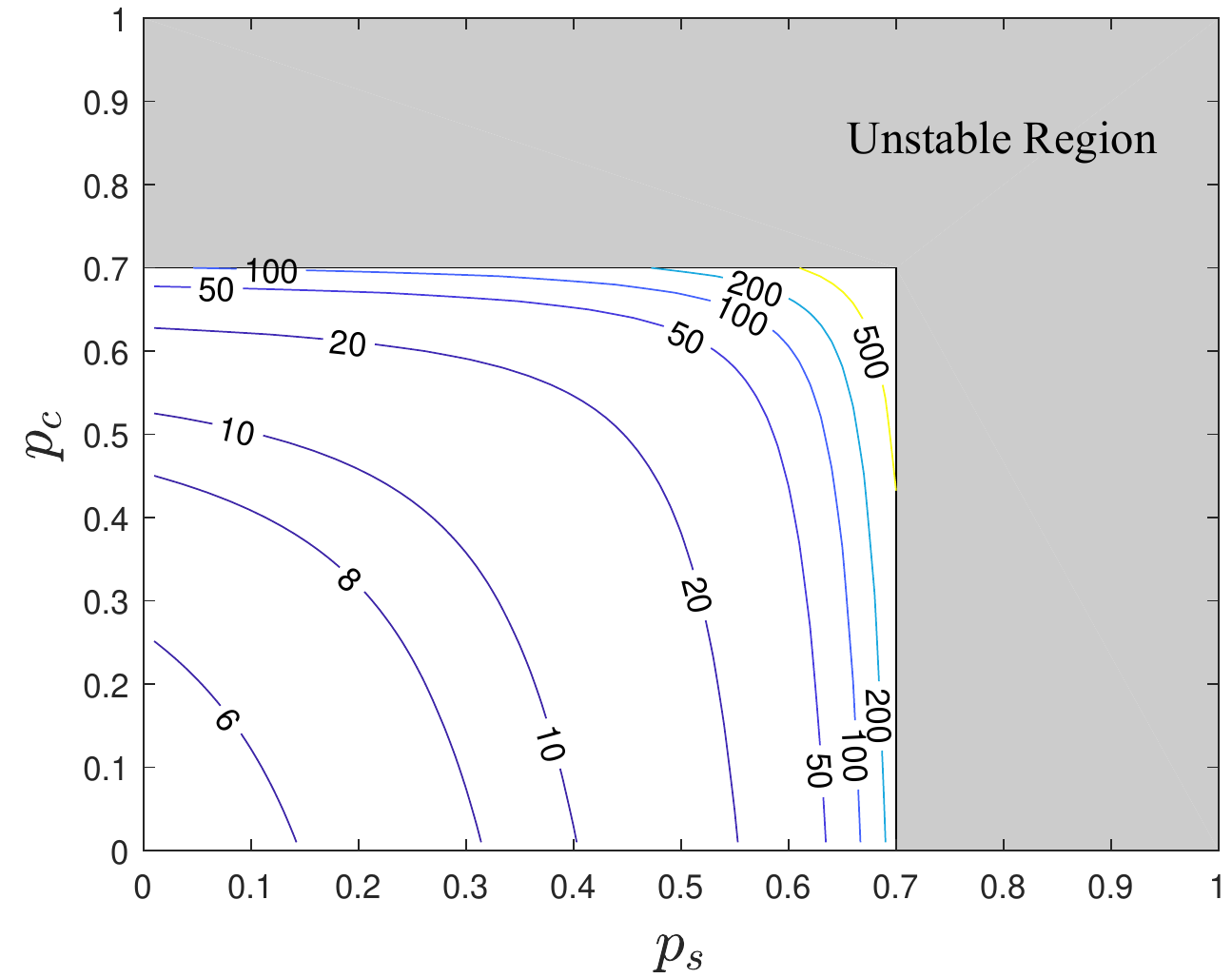}
	\vspace{-0.3cm}	
	\caption{One-step controllable case: average cost versus packet-error probabilities, i.e., $p_s$ and $p_c$.}
	\label{fig:average_cost}
	\vspace{-0.5cm}
\end{figure}

\emph{\underline{Performance versus transmission reliabilities.}} 
In Fig.~\ref{fig:average_cost}, we show a contour plot of the average cost of the plant with different uplink-downlink packet-error probabilities $(p_s,p_c)$ within the rectangular region that can stabilize the plant, i.e., $p_s,p_c < 1/\rho^2(\mathbf{A}) =0.7$ based on Theorem~\ref{theorem:existence}.
The average cost is calculated by running a $10^6$-step simulation and then taking the average, and it does not have a steady-state value if $(p_s,p_c)$ lies outside the rectangular region.
We see that the average cost increases quickly when $p_s$ or $p_c$ approaches the boundary $1/\rho^2(\mathbf{A})$. Also, it is interesting to see that in order to guarantee a certain average cost, e.g., $J=8$, the required $p_s$ is less than $p_c$ in general, which implies that the transmission reliability of the sensor-controller channel is more important than that of the controller-actuator channel.
%
%
%the average cost of $p_s=0$ and $p_c=p$ is smaller than that of $p_c=0$ and $p_s=p$, which may imply that to enhance the performance of a WNCS, improving the transmission reliability of the sensor-controller channel is more important than the controller-actuator channel.

{\emph{\underline{Fading channel scenario.}} Assume that both the uplink and downlink channels have two Markov channel states with the packet error probabilities $0.1$ (i.e., the good channel state) and $0.4$ (i.e., the bad channel state), respectively, i.e., $\omega_1=\xi_1=0.1$ and $\omega_2=\xi_2=0.4$.
Figs.~\ref{fig:markov1} and~\ref{fig:markov2} show the average cost versus the simulation time with different channel state transition probabilities.
In Fig.~\ref{fig:markov1}, we set the matrices of the channel state transition probabilities of the uplink and downlink channels as $\mathbf{D}_s = \mathbf{D}_c = \begin{bmatrix}
0.5 & 0.5 \\
0.5 & 0.5
\end{bmatrix}$. Taking the uplink channel as an example, the transition probabilities from the bad channel state and the good channel state are the same, and thus the Markov channel does not have any memory~\cite{gupta2006receding}. Since the uplink and downlink channels have the same Markovian property, both the uplink and downlink channels are memoryless. 
In Fig.~\ref{fig:markov2}, we set $\mathbf{D}_s = \mathbf{D}_c = \begin{bmatrix}
0.8 & 0.2 \\
0.2 & 0.8
\end{bmatrix}$, where the probability of remaining in any given state is higher than jumping to the other state.
In this case, both the uplink and downlink channels have persistent memories. 

In Figs.~\ref{fig:markov1} and~\ref{fig:markov2}, we see that the persistent policy always provides a low average cost, which is close to that of the optimal policy and is much smaller than that of the naive policy.
It is interesting to see that the average cost achieved by the optimal policy under the memoryless Markov channels in Fig.~\ref{fig:markov1} is smaller than that of the Markov channels with memories in Fig.~\ref{fig:markov2}.
This is because in the Markov channel with memory, when the current channel state is bad, it is more likely to have a bad channel state in the following time slot, which can lead to consecutive packet losses and deteriorate the control performance of the WNCS. }

\begin{figure}[t]
	\centering
	\includegraphics[scale=0.58]{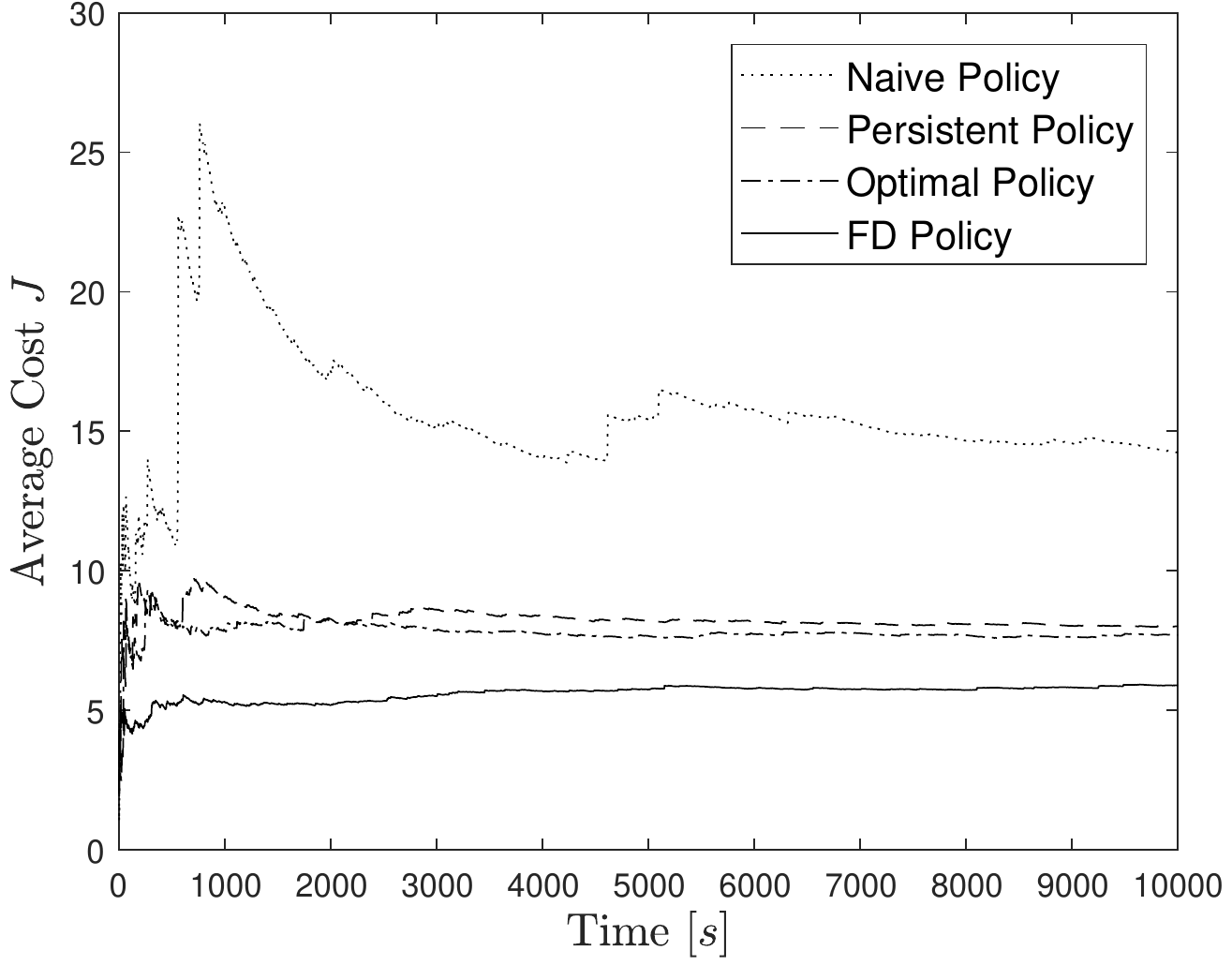}
	\vspace{-0.3cm}	
	\caption{Markov channel scenario without memory: average cost versus time.}
	\label{fig:markov1}
	\vspace{-0.0cm}
\end{figure}
\begin{figure}[t]
	\centering
	\includegraphics[scale=0.58]{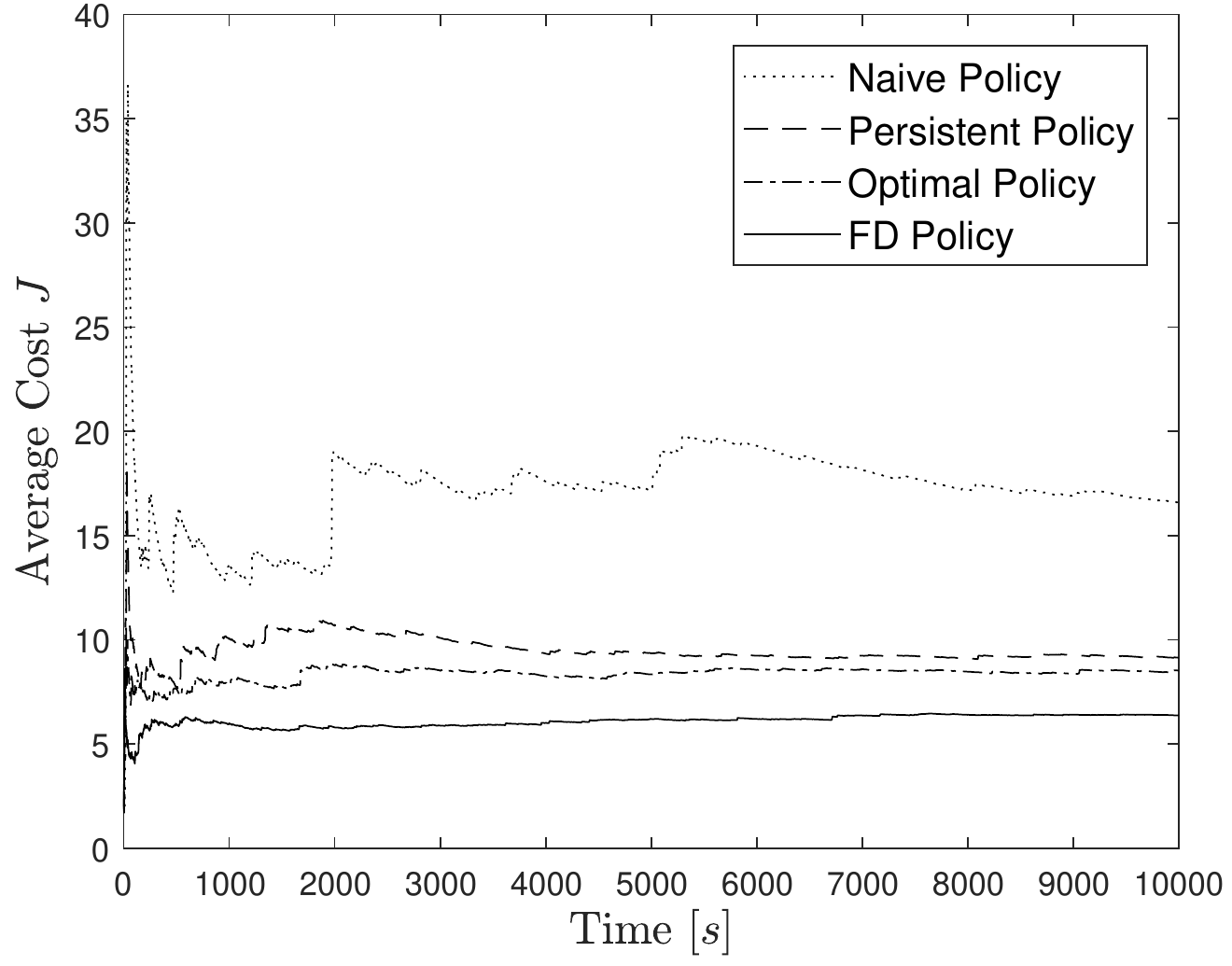}
	\vspace{-0.3cm}	
	\caption{Markov channel scenario with memory: average cost versus time.}
	\label{fig:markov2}
	\vspace{-0.0cm}
\end{figure}

\subsection{Two-Step Controllable Case} \label{single1}
In this case, we assume that $\mathbf{B} = -\left[1,1\right]^{\top}$, and $\mathbf{K}=\left[2.9,-1\right]$ satisfying $\left(\mathbf{A+BK}\right)^2=\mathbf{0}$. For fair comparison, all the policies considered in this subsection adopt the same predictive control method in~\eqref{C}, \eqref{U} and \eqref{u} with $v=2$.

In Fig.~\ref{fig:DeadBeat_v2}, we plot the average cost function versus the packet-error probability of the downlink channel with different uplink-downlink transmission-scheduling policies, where the uplink packet-error probability $p_s=0.1$. We see that the persistent policy can still provide a good performance close to the optimal policy. Given the FD policy as a benchmark, it is clear that the optimal scheduling policy provides at least a $66\%$ reduction of the average cost than the naive policy when $p_c\geq 0.1$.

%\begin{figure*}
%	\begin{minipage}{0.5\textwidth}
%	\centering
%	\includegraphics[scale=0.58]{DeadBeat_v2.pdf}
%	\vspace{-0.5cm}	
%	\caption{Two-step controllable case: average cost versus packet-error probability $p_c$.}
%	\label{fig:DeadBeat_v2}
%	\vspace{-0.5cm}
%\end{minipage}
%%
%\begin{minipage}{0.5\textwidth}
%	\centering
%	\includegraphics[scale=0.58]{LQG_v2_v3_comparison.pdf}
%	\vspace{-0.5cm}	
%	\caption{Non-finite-step-controllable case: average cost versus packet-error probability $p_c$.}
%	\label{fig:LQG_v2_v3_comparison}
%	\vspace{-0.5cm}
%\end{minipage}
%\end{figure*}

\begin{figure}[t]
	\centering
	\includegraphics[scale=0.58]{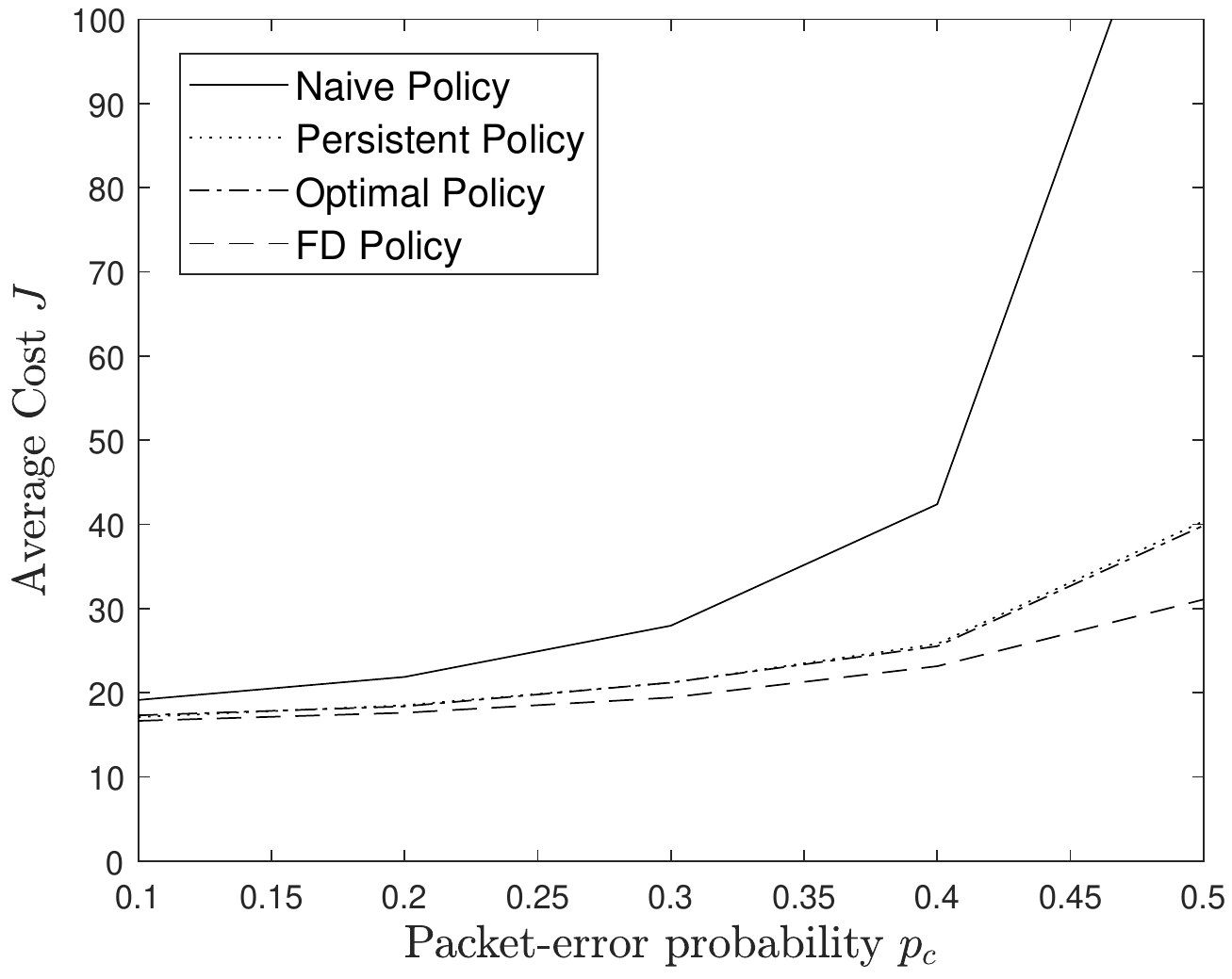}
	\vspace{-0.3cm}	
	\caption{Two-step controllable case: average cost versus packet-error probability $p_c$.}
	\label{fig:DeadBeat_v2}
	\vspace{-0.4cm}
\end{figure}

\begin{figure}[t]
	\centering
	\includegraphics[scale=0.58]{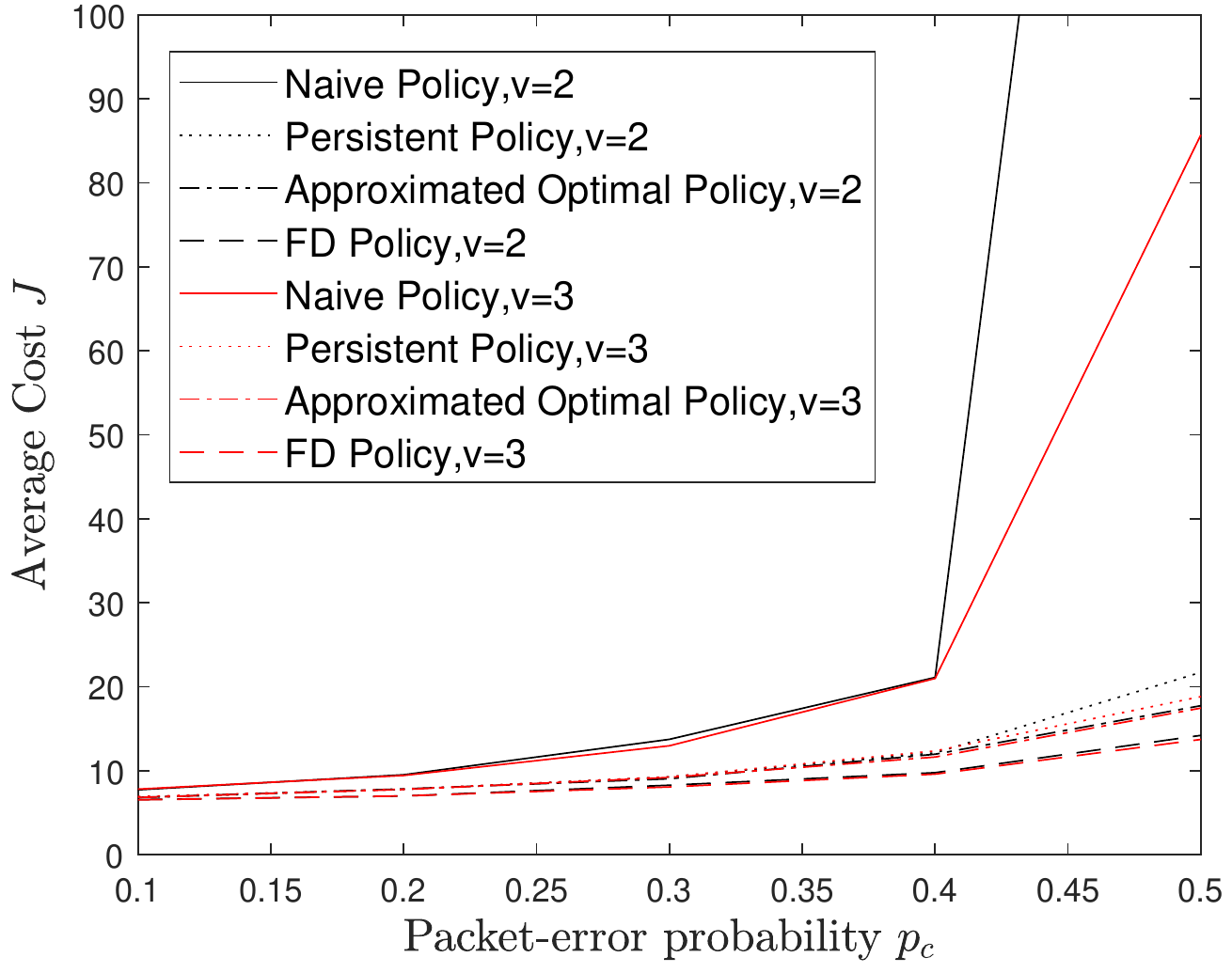}
	\vspace{-0.3cm}	
	\caption{Non-finite-step-controllable case: average cost versus packet-error probability $p_c$.}
	\label{fig:LQG_v2_v3_comparison}
	\vspace{-0.5cm}
\end{figure}

\subsection{Non-Finite-Step-Controllable Case}\label{single2}
We now look at the non-finite-step-controllable case as discussed in Remark~\ref{remark:extension}, where  $\mathbf{B} =- \left[1,1\right]^{\top}$, and $\mathbf{K}=\left[0.7,0.4\right]$. It can be verified that $\rho(\mathbf{A+BK})=0.72<1$ and $(\mathbf{A+BK})^v\neq \mathbf{ 0}$ for a practical range of $v$, e.g. $v<10$. 
We consider two predictive control protocols in \eqref{C} with $v=2$ and $v=3$, respectively, i.e., the controller sends two or three commands to the actuator each time.
We have $(\mathbf{A+BK})^1=\begin{bmatrix}
0.4 & -0.2 \\
-0.5 & 0.4
\end{bmatrix}$,
$(\mathbf{A+BK})^2=\begin{bmatrix}
0.26 & -0.16 \\
-0.4 & 0.26
\end{bmatrix}$
and $(\mathbf{A+BK})^3=\begin{bmatrix}
0.18 & -0.12 \\
-0.29 & 0.18
\end{bmatrix}$. It is clear that $(\mathbf{A+BK})^v$ approaches $\mathbf{0}$ as $v$ increases.
By letting $(\mathbf{A+BK})^v =\mathbf{0}$ in the analysis of plant-state vector in \eqref{detail_representaiton}, where $v=2$ or $3$, the plant-state covariance matrix $\mathbf{P}_k$ is approximated by a function of $2 v-1$ state parameters as in Proposition~\ref{covariance_analysis}. Based on such the approximation, we can formulate and solve the MDP problem in Section~\ref{sec:mdp2}, resulting an \emph{approximated optimal scheduling policy}.

In Fig.~\ref{fig:LQG_v2_v3_comparison},
we plot the average cost function versus the packet-error probability of the downlink channel with different downlink transmission-scheduling policies.
We see that for both the cases $v=2$ and $3$, the performance of the approximated optimal and persistent uplink-downlink scheduling policies are quite close to the benchmark FD policy when $p_c<0.2$, while the performance gap between the naive scheduling policy and the FD policy is large. This also implies that the approximated optimal policy is near optimal in this practical range of downlink transmission reliability, and the persistent scheduling policy is also an effective yet low-complexity one in this case.

\section{Conclusions}\label{sec:con}
{In this work, we have proposed an important uplink-downlink transmission scheduling problem of a WNCS with an HD controller for practical IIoT applications, which has not been considered in the open literature.
We have given a comprehensive analysis of the estimation-error covariance and the plant-state covariance of the HD-controller-based WNCS for both one-step and $v$-step controllable plants.
Based on our analytical results, in both the static and fading channel scenarios, we have formulated 
the novel problem to optimize the transmission-scheduling policy depending on both the current estimation quality of the controller and the current cost function of the plant, 
so as to minimize the long-term average cost function.
Moreover, for the static channel scenario, we have derived the necessary and sufficient condition of the existence of a stationary and deterministic optimal policy that results in a bounded average cost in terms of the transmission reliabilities of the uplink and downlink channels. For the fading channel scenario, we have derived a necessary condition and a sufficient condition in terms of the uplink and downlink channel qualities, under which the optimal transmission-scheduling policy exits.
Our problem can be solved effectively by the standard MDP algorithms if the optimal scheduling policy exits.
Also, we have derived an easy-to-compute suboptimal policy, which provides a control performance close to the optimal policy and notably reduces the average cost of the plant compared to a naive alternative-scheduling policy.

In future work, we will consider a scenario that an HD controller controls multiple plants for IIoT applications with a large number of devices. It is important to investigate the scheduling of different sensors' transmissions to the controller and the controller's transmissions to different actuators, and consider different quality of service (QoS) requirements of different devices in the scheduling and how they affect the control. Moreover, for the scheduling-policy design, it is more practical to take into account the transmission power constraints of the sensors and the controller.}

\section*{Appendix A: Proof of Proposition~\ref{covariance_analysis}}
Recall that $\eta^0_k \triangleq \eta_{k-1}$. From the definition of $\eta_k$ in \eqref{eta}, we have 
%$\eta^0_{j+1} =\eta^0_{j}+1$ when $j=k-\eta^0_k,\cdots,k-1$.
\begin{equation}
\eta^0_{j+1}=\begin{cases}
1, & j=k-\eta^0_k\\
\eta^0_{j}+1, & j=k-\eta^0_k+1,\cdots,k-1
\end{cases}
\end{equation}
By using the state-updating rule \eqref{v_x} for $\mathbf{x}_j$, $j=(k-\eta^0_k+1),\cdots,k$, we have
\begin{equation}
\left\lbrace
\begin{aligned}
\mathbf{x}_{k-\eta^0_k+1}&= \mathbf{A} \mathbf{x}_{k-\eta^0_k} + \mathbf{BK(A+BK)}^0 \hat{\mathbf{x}}_{k-\eta^0_k} +\mathbf{w}_{k-\eta^0_k}\\
\mathbf{x}_{k-\eta^0_k+2}&= \mathbf{A} \mathbf{x}_{k-\eta^0_k+1} + \mathbf{BK(A+BK)}^1 \hat{\mathbf{x}}_{k-\eta^0_k} +\mathbf{w}_{k-\eta^0_k+1}\\
&\vdots\\
\mathbf{x}_{k}&= \mathbf{A} \mathbf{x}_{k-1} + \mathbf{BK(A+BK)}^{\eta^0_k-1} \hat{\mathbf{x}}_{k-\eta^0_k} +\mathbf{w}_{k-1}\\
\end{aligned}\right.
\end{equation}
Substituting  $\mathbf{x}_{k-\eta^0_k+1}$ into $\mathbf{x}_{k-\eta^0_k+2}$ and so on, it can be shown that 
\begin{equation} \label{state_representation}
\begin{aligned}
\mathbf{x}_{k} 
&= (\mathbf{A}+\mathbf{B}\mathbf{K})^{\eta^0_k}\mathbf{x}_{k-\eta^0_k} + (\mathbf{A}^{\eta^0_k}-(\mathbf{A}+\mathbf{B}\mathbf{K})^{\eta^0_k})\mathbf{e}_{k-\eta^0_k} \\
&+ \sum_{i=1}^{\eta^0_k}\mathbf{A}^{i-1}\mathbf{w}_{k-i}.
\end{aligned}
\end{equation}
Using the new state-updating rule \eqref{state_representation}, $\mathbf{x}_{k}$ can be further rewritten as
\begin{equation} \notag
\begin{aligned}
&\mathbf{x}_{k} \!=\! (\mathbf{A}\!+\!\mathbf{B}\mathbf{K})^{\eta_k^0}\mathbf{x}_{t^1_k} \!+\! (\mathbf{A}^{\eta_k^0}\!-\!(\mathbf{A}\!+\!\mathbf{B}\mathbf{K})^{\eta_k^0})\mathbf{e}_{t^1_k} \!+\! \!\sum_{i=1}^{\eta_k^0}\!\mathbf{A}^{i-1}\mathbf{w}_{k-i}\\
&=\! (\mathbf{A}\!+\!\mathbf{B}\mathbf{K})^{\eta_k^0}\big((\mathbf{A}\!+\!\mathbf{B}\mathbf{K})^{\eta_k^1}\mathbf{x}_{t^2_k} \!+\! (\mathbf{A}^{\eta_k^1}\!-\!(\mathbf{A}+\mathbf{B}\mathbf{K})^{\eta_k^1})\mathbf{e}_{t^2_k} \\
&+ \sum_{i=1}^{\eta_k^1}\mathbf{A}^{i-1}\mathbf{w}_{t^1_k-i}\big)+ (\mathbf{A}^{\eta_k^0}-(\mathbf{A}+\mathbf{B}\mathbf{K})^{\eta_k^0})\mathbf{e}_{t^1_k} + \sum_{i=1}^{\eta_k^0}\mathbf{A}^{i-1}\mathbf{w}_{k-i}\\
&=\! (\mathbf{A}\!+\!\mathbf{B}\mathbf{K})^{\eta_k^0\!+\!\eta_k^1}\mathbf{x}_{t^2_k} \!+\!\! \sum_{i=1}^{\eta_k^0}\!\mathbf{A}^{i-1}\mathbf{w}_{k-i} \!+\! (\mathbf{A}\!+\!\mathbf{B}\mathbf{K})^{\eta_k^0}\sum_{i=1}^{\eta_k^1}\!\mathbf{A}^{i\!-\!1}\mathbf{w}_{t^1_k-i}\\
\end{aligned}
\end{equation}
\begin{equation} \label{detail_representaiton}
\begin{aligned}
&+\! (\mathbf{A}^{\eta_k^0}\!-\!(\mathbf{A}\!+\!\mathbf{B}\mathbf{K})^{\eta_k^0})\mathbf{e}_{t^1_k} \!+\! (\mathbf{A}\!+\!\mathbf{B}\mathbf{K})^{\eta_k^0}(\mathbf{A}^{\eta_k^1}\!-\!(\mathbf{A}\!+\!\mathbf{B}\mathbf{K})^{\eta_k^1})\mathbf{e}_{t^2_k}\\
&= (\mathbf{A}+\mathbf{B}\mathbf{K})^{\eta_k^0+\eta_k^1+\cdots + \eta_k^{v-1}}\mathbf{x}_{t^v_k} + \mathbf{w}'+\mathbf{e}'\\
&= \mathbf{w}'+\mathbf{e}',
\end{aligned}
\end{equation}
where the last step is due to the fact that
$\eta_k^0+\eta_k^1+\cdots + \eta_k^{v-1} \geq v$ as $\eta_k^i \geq 1, \forall i \geq 0$, and $(\mathbf{A}+\mathbf{B}\mathbf{K})^v = \mathbf{0}$, and
\begin{equation}
\begin{aligned}
\mathbf{w}'&=\sum_{i=1}^{\eta_k^0}\mathbf{A}^{i-1}\mathbf{w}_{k-i}+(\mathbf{A}+\mathbf{B}\mathbf{K})^{\eta_k^0}\sum_{i=1}^{\eta_k^1}\mathbf{A}^{i-1}\mathbf{w}_{k-i} \\
&+ \cdots + (\mathbf{A}+\mathbf{B}\mathbf{K})^{\eta_k^0+\cdots + \eta_k^{v-2}}\sum_{i=1}^{\eta_k^{v-1}}\mathbf{A}^{i-1}\mathbf{w}_{t^{v-1}_k-i},
\end{aligned}
\end{equation}
\begin{equation} \label{e_i}
\mathbf{e}_{t^j_k} = \sum_{i=1}^{\tau_k^{j}}\mathbf{A}^{i-1}\mathbf{w}_{t^{j}_k-i}, j=1,\cdots,v,
\end{equation}
\begin{equation}
\begin{aligned}
&\mathbf{e}'=(\mathbf{A}^{\eta_k^0}-(\mathbf{A}+\mathbf{B}\mathbf{K})^{\eta_k^0})\mathbf{e}_{t^1_k}\\
&+(\mathbf{A}+\mathbf{B}\mathbf{K})^{\eta_k^0}(\mathbf{A}^{\eta_k^1}-(\mathbf{A}+\mathbf{B}\mathbf{K})^{\eta_k^1})\mathbf{e}_{t^2_k} 
+ \cdots + \\ &(\mathbf{A}+\mathbf{B}\mathbf{K})^{\eta_k^0+\cdots + \eta_k^{v-2}}(\mathbf{A}^{\eta_k^{v-1}}-(\mathbf{A}+\mathbf{B}\mathbf{K})^{\eta_k^{v-1}})\mathbf{e}_{t^v_k}.
\end{aligned}
\end{equation}
We see that $\mathbf{x}_{k}$ only depends on the noise terms in the time range \begin{equation}
\mathcal{S} \triangleq \left[k-(\eta_k^0+\cdots+\eta_k^{v-1})-\tau_v,k-1\right].
\end{equation}

%In order to calculate the covariance of $\mathbf{x}_{k}$,$\mathbf{P}_{k}$, we need to clarify the coherent relationship between $\mathbf{e}_i$ and $\mathbf{w}_j$.

To further simplify \eqref{detail_representaiton}, we consider three complementary cases: 1) $\tau^i_k < \eta^i_k, \forall i = 1,\cdots, v-1$, i.e., a sensor's successful transmission occurred between two consecutive controller's successful transmissions, as illustrated in Fig.~\ref{fig:time2}(a); 2) there exists $i$ such that $\tau^i_k \geq \eta^i_k$ and there also exists $j$ such that $\tau^j_k<\eta^j_k$ where $i,j \in \{1,\cdots, v-1\}$, i.e., a sensor's successful transmission did not always occur between two consecutive controller's successful transmissions, as illustrated in Fig.~\ref{fig:time2}(b). Note that from the definition of $\tau^j_k$ and $\eta^j_k$, $\tau^{i}_k =\eta^i_k + \tau^{i+1}_k $ if $\tau^{i}_k>\eta^{i}_k$; 3) $\tau^i_k = \eta^i_k + \tau^{i+1}_k  \geq \eta^i_k$ for all $i \in \{0,\cdots, v-1\}$, i.e., a sensor's successful transmission never occur between the first and the $v$th controller's successful transmissions prior to the current time slot~$k$, as illustrated in Fig.~\ref{fig:time2}(c).

\begin{figure}[t]
	\centering
	\includegraphics[scale=0.68]{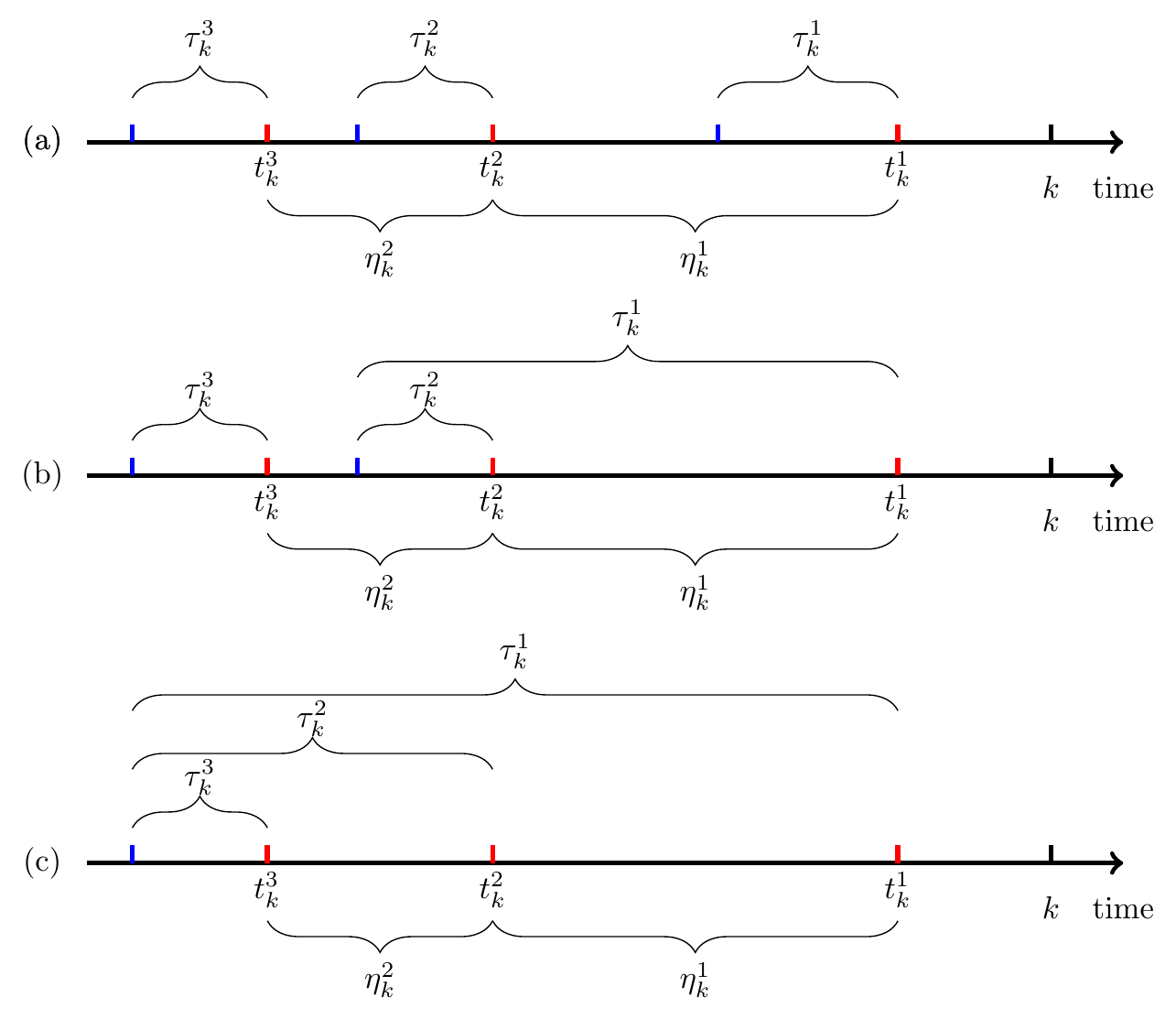}
	\vspace{-0.5cm}	
	\caption{Illustration of three different cases for analyzing the plant-state covariance, where red vertical bars denote successful controller's transmissions and blue vertical bars denote the most recent successful sensor's transmissions prior to the successful controller's transmissions.}
	\label{fig:time2}
	\vspace{-0.5cm}
\end{figure}

For case 1), 
$\mathbf{e}_{t^i_k}$ contains the noise terms within time slots $t^i_k-\tau^i_k$ to $t^i_k-1$. Since $\tau^i_k < \eta^i_k = t^{i+1}_k -t^{i}_k$, $\mathbf{e}_{t^i_k}$ and  $\mathbf{e}_{t^j_k}$ do not contain common noise terms when $i\neq j$.
%, and thus each $\mathbf{e}_i$ is independent with the others.
Taking \eqref{e_i} into \eqref{detail_representaiton}, after some simple simplifications, $\mathbf{x}_{k}$ can be simplified as below with $v$-segment summations
\begin{equation} \label{state_expression_final}
\begin{aligned}
\mathbf{x}_{k} 
%&= \sum_{i=1}^{\eta_k^0}\mathbf{A}^{i-1}\mathbf{w}_{k-i} + \mathbf{A}^{\eta_k^0}\sum_{i=1}^{\tau_k^1}\mathbf{A}^{i-1}\mathbf{w}_{t_1-i} + (\mathbf{A}+\mathbf{B}\mathbf{K})^{\eta_k^0}\sum_{i=\tau_k^1+1}^{\eta_1}\mathbf{A}^{i-1}\mathbf{w}_{t_1-i} + \cdots\\
%&+ (\mathbf{A}+\mathbf{B}\mathbf{K})^{\eta_k^0+\cdots+\eta_k^{v-3}}\mathbf{A}^{\eta_k^{v-2}}\sum_{i=1}^{\tau_k^{v-1}}\mathbf{A}^{i-1}\mathbf{w}_{t_{v-1}-i} + (\mathbf{A}+\mathbf{B}\mathbf{K})^{\eta_k^0+\cdots+\eta_k^{v-2}}\sum_{i=\tau_k^{v-1}+1}^{\eta_k^{v-1}}\mathbf{A}^{i-1}\mathbf{w}_{t_{v-1}-i}\\
%&+ (\mathbf{A}+\mathbf{B}\mathbf{K})^{\eta_k^0+\cdots + \eta_k^{v-2}}(\mathbf{A}^{\eta_k^{v-1}}-(\mathbf{A}+\mathbf{B}\mathbf{K})^{\eta_k^{v-1}})\sum_{i=1}^{\tau_k^{v}}\mathbf{A}^{i-1}\mathbf{w}_{t_{v}-i}\\
&= \sum_{i=1}^{\eta_k^0+\tau_k^1}\mathbf{A}^{i-1}\mathbf{w}_{k-i} + (\mathbf{A}+\mathbf{B}\mathbf{K})^{\eta_k^0}\sum_{i=\tau_k^1+1}^{\eta_k^1+\tau_k^2}\mathbf{A}^{i-1}\mathbf{w}_{t^1_k-i} \\
&+ \cdots + (\mathbf{A}+\mathbf{B}\mathbf{K})^{\eta_k^0+\cdots+\eta_k^{v-2}}\sum_{i=\tau_k^{v-1}+1}^{\eta_k^{v-1}+\tau_k^v}\mathbf{A}^{i-1}\mathbf{w}_{t^{v-1}_k-i},
\end{aligned}
\end{equation}
where $\eta^j_k > \tau^{j}_k,\forall j=1,\cdots,v-1$.

%The system noises $\{\mathbf{w}_1,\cdots, \mathbf{w}_k\}$ is assumed to be independent and identity distributed.

%\begin{equation} \label{state_covariance}
%\begin{aligned}
%\mathbf{P}_{k} &=  \sum_{i=1}^{\eta_k^0+\tau_k^1}\mathbf{A}^{i-1}\mathbf{R}(\mathbf{A}^\top)^{i-1} + (\mathbf{A}+\mathbf{B}\mathbf{K})^{\eta_k^0}\sum_{i=\tau_k^1+1}^{\eta_k^1+\tau_k^2}\mathbf{A}^{i-1}\mathbf{R}(\mathbf{A}^\top)^{i-1}\big((\mathbf{A}+\mathbf{B}\mathbf{K})^{\eta_k^0}\big)^\top\\
%&+ \cdots + (\mathbf{A}+\mathbf{B}\mathbf{K})^{\eta_k^0+\cdots+\eta_k^{v-2}}\sum_{i=\tau_k^{v-1}+1}^{\eta_k^{v-1}+\tau_k^v}\mathbf{A}^{i-1}\mathbf{R}(\mathbf{A}^\top)^{i-1}\big((\mathbf{A}+\mathbf{B}\mathbf{K})^{\eta_k^0+\cdots+\eta_k^{v-2}}\big)^\top
%\end{aligned}
%\end{equation}
%
%From \eqref{state_covariance}, the state covariance $\mathbf{P}_{k}$ can be exactly determined by variables $\{\eta_k^0,\eta_k^1,\cdots,\eta_k^{v-1}\}$ and $\{\tau_k^1,\tau_k^2,\cdots,\tau_k^v\}$.

For case 2), the estimation-error terms $\mathbf{e}_{t^i_k}$ and $\mathbf{e}_{t^j_k}$ in \eqref{detail_representaiton} may contain common noise terms when $i\neq j$, and $\eta^j_k > \tau^{j}_k$ may not hold for $j=1,\cdots,v-1$.
Inspired by the result \eqref{state_expression_final} in the first case, to calculate $\mathbf{x}_{k}$, we divide the time range $\mathcal{S}$ by the time slots
$t^j_k-\tau^j_k$, $j=1,\cdots,v-1$. Since $t^{j'}_k-\tau^{j'}_k$ may equal to $t^j_k-\tau^j_k$ when $j'\neq j$, $\mathcal{S}$ is divided into $v'$ segments from left to right, and $1 \leq v'\leq v$.
%
% of successful sensor's transmissions between $t^{v}_k$ and $t^1_k$. Let $v'$ denote the number of successful sensor's transmissions in the range, where $0 \leq v'<v-1$. Thus, $\mathcal{S}$ is divided into $v'+1$ segments.

To investigate the noise terms within the first $v'-1$ segments of $\mathcal{S}$, we assume that
sensor's successful transmissions occurred in the time ranges $\left[t^{j'+1}_k+1,t^{j'}_k\right]$ and $\left[t^{j+1}_k+1,t^j_k\right]$ and there is no sensor's successful transmission in the gap between them, where
$v \geq j'>j\geq 1$. Thus, $\eta^{j'}_k>\tau^{j'}_k$  and $\eta^{j}_k>\tau^{j}_k$.
When $j'=j+1$, we have $t^{j'}_k-\tau^{j'}_k=t^{j+1}_k-\tau^{j+1}_k=t^{j}_k-\eta^j_k-\tau^{j+1}_k$ and only $\mathbf{w}'$ and the estimation-error term $\mathbf{e}_{t^{j'}_k}$ contains the noise terms within the time segment $\left[ t^{j'}_k-\tau^{j'}_k,t^j_k-\tau^j_k -1\right]= \left[ t^j_k-(\eta^j_k+\tau^{j+1}_k),t^j_k-\tau^j_k -1\right]$, therefore, the noise terms in this segment have exactly the same expressions as in \eqref{state_expression_final} of case 1), i.e., 
\begin{equation} \label{same_terms}
(\mathbf{A}+\mathbf{B}\mathbf{K})^{\eta_k^0+\cdots+\eta_k^{j-1}}\sum_{i=\tau_k^{j}+1}^{\eta_k^{j}+\tau_k^{j+1}}\mathbf{A}^{i-1}\mathbf{w}_{t^{j}_k-i}.
\end{equation}
When $j'>j+1$, $\mathbf{w}'$ and  the estimation-error terms $\mathbf{e}_{t^{j+1}_k},\mathbf{e}_{t^{j+2}_k},\cdots,\mathbf{e}_{t^{j'}_k}$ contains the noise terms within the time segment $\left[ t^{j'}_k-\tau^{j'}_k,t^j_k-\tau^j_k-1\right]=\left[ t^j_k-(\eta^j_k+\tau^{j+1}_k),t^j_k-\tau^j_k -1\right]$.
After combining the noise terms in this range, we also have the expression \eqref{same_terms}.

To investigate the noise terms of the $v'$th (last) segment of $\mathcal{S}$, we assume that the most recently successful sensor's transmission before $t^1_k$ is within the range of $\left[t^{j+1}_k+1,t^j_k \right]$, where $j\in \{1,\cdots,v\}$.
We see that $\mathbf{w}'$ and the estimation-error terms $\mathbf{e}_{t^{1}_k},\cdots,\mathbf{e}_{t^{j}_k}$ contains the noise terms within the time range $\left[ t^{j}_k-\tau^{j}_k,k-1\right]=\left[ k-(\eta^0_k+\tau^{1}_k),k-1\right]$. After combining the noise terms in this range contributed by $\mathbf{e}_{t^{1}_k},\cdots,\mathbf{e}_{t^{j}_k}$ and $\mathbf{w}'$, we have exactly the same expressions as in \eqref{state_expression_final} of case 1), i.e.,
\begin{equation} \label{last segment}
\sum_{i=1}^{\eta_k^0+\tau_k^1}\mathbf{A}^{i-1}\mathbf{w}_{k-i}.
\end{equation}

To sum up, different from \eqref{state_expression_final} of case 1), $\mathbf{x}_{k}$ of case 2) has $v'$ segment summations, i.e.,
\begin{equation} \label{state_expression_final3}
\begin{aligned}
&\mathbf{x}_{k} 
\!\!=\!\!\! \sum_{i=1}^{\eta_k^0\!+\!\tau_k^1}\!\!\mathbf{A}^{i-1}\mathbf{w}_{k-i} \!+\! \mathbbm{1}(\eta^1_k\!>\!\tau^1_k) (\mathbf{A}\!+\!\mathbf{B}\mathbf{K})^{\eta_k^0}\!\!\sum_{i=\tau_k^1+1}^{\eta_k^1+\tau_k^2}\!\!\mathbf{A}^{i-1}\mathbf{w}_{t^1_k-i} \\
&+ \cdots+ \\
& \mathbbm{1}(\eta^{v-1}_k\!>\!\tau^{v-1}_k)(\mathbf{A}+\mathbf{B}\mathbf{K})^{\eta_k^0+\cdots+\eta_k^{v-2}}\!\!\sum_{i=\tau_k^{v-1}\!+\!1}^{\eta_k^{v-1}\!+\!\tau_k^v}\!\!\!\!\mathbf{A}^{i-1}\mathbf{w}_{t^{v-1}_k-i},
\end{aligned}
\end{equation}
where $\mathbbm{1}(\cdot)$ is the indicator function and $\sum_{j=1}^{v-1}\mathbbm{1}(\eta^j_k>\tau^j_k)=v'-1$.

For case 3), the range $\mathcal{S}$ has only one segment, which is a special case of case 2) discussed above \eqref{last segment}, where $j=v$. Therefore, $\mathbf{x}_{k}$ has the expression of \eqref{last segment}.

Therefore, the general expression of $\mathbf{x}_k$ is given in \eqref{state_expression_final3}, and thus the state covariance $\mathbf{P}_{k} = \mathbb{E}[\mathbf{x}_{k}\mathbf{x}_{k}^\top]$ is obtained as \eqref{one-state_covariance0}.
%, where $\phi^j_k \triangleq \eta^j_k +\tau^{j+1}_k$, $\forall j= 0,\cdots,v-1$.

\section*{Appendix B: Proof of Theorem~\ref{theorem:existence}}
We prove that the stationary and deterministic policy $\pi'$ in \eqref{suboptimal} stabilizes the plant.

It is easy to verify that the state-transition process induced by $\pi'$ is an ergodic Markov process, i.e., any state in $\mathbb{S}$ is aperiodic and positive recurrent. In the following, we prove that the average cost of the plant induced by $\pi'$ is bounded.

From \eqref{suboptimal}, \eqref{tau} and \eqref{phi}, we see that the policy $\pi'$ is actually~a \emph{persistent scheduling policy}, which consecutively schedules the uplink transmission until a transmission is successful and then consecutively schedules the downlink transmission until a transmission is successful, and so on. The transmission process of (s)ensor's measurement and (c)ontroller's command is illustrated~as 
\begin{equation}\label{process}
\{\cdots, \overbrace{\underbrace{\text{s} \cdots \text{s}}_{m'},\underbrace{\text{c} \cdots \text{c}}_{n'}}^{\text{control cycle } (t-1)}, \overbrace{\underbrace{\text{s} \cdots \text{s}}_{m},\underbrace{\text{c} \cdots \text{c}}_{n}}^{\text{control cycle } t},\cdots \}
\end{equation}
where $m$ and $n$ are the numbers of consecutively scheduled uplink and downlink transmission, respectively. 

For the ease of analysis, we define the concept of \emph{control cycle}, which consists of $M$ consecutive uplink transmissions and the following $N$ consecutive downlink transmissions. It is clear that $M$ and $N$ follow geometric distributions with success probabilities $(1-p_s)$ and $(1-p_c)$, respectively. The values of $M$ and $N$ change in different control cycles independently as illustrated in \eqref{process}. Thus, the uplink-downlink schedule process \eqref{process} can be treated as a sequence of control cycles.

Let $S$ and $L\triangleq M+N$ denote the sum cost of the plant and the number of transmissions in a control cycle, respectively. 
We can prove that $S$ and $L$ of the sequence of control cycles can be treated as ergodic Markov chains, i.e., $\{\cdots, S_t,S_{t+1},\cdots\}$ and  $\{\cdots, L_t,L_{t+1},\cdots\}$, where $t$ is the control-cycle index.
We use $N'$ to denote the number of consecutive downlink transmissions before the current control cycle, which follows the same distribution of $N$.
Due to the ergodicity of $\{S_t\}$ and $\{L_t\}$, the average cost in \eqref{metric} can be rewritten~as 
\begin{align}\label{J}
J =\lim\limits_{t \rightarrow \infty} \frac{S_1+S_2+\cdots+S_t}{L_1+L_2+\cdots+L_t}
= \frac{\mathbb{E}\left[S\right]}{\mathbb{E}\left[L\right]},
\end{align}
where 
\begin{align} \label{S}
\mathbb{E}\left[S\right] &= \sum_{n'=1}^{\infty}\sum_{m=1}^{\infty}\sum_{n=1}^{\infty}\mathbb{E}\left[S|N' = n', M = m, N = n\right]\\ \notag
&\mathbb{P}[N'=n',M=m,N=n],\\ \label{L}
\mathbb{E}\left[L\right] &= \sum_{n'=1}^{\infty}\sum_{m=1}^{\infty}\sum_{n=1}^{\infty}\left(m+n\right)\mathbb{P}[N'=n',M=m,N=n].
\end{align}

Thus, the average cost $J$ is bounded if $\mathbb{E}\left[S\right]$ is.
From the policy \eqref{suboptimal} and the state-transition rules in \eqref{tau} and \eqref{phi}, we see that $\phi$ is equal to $N'+1$ at the beginning of the control cycle, and increases one-by-one within the control cycle, and we have 
\begin{equation}
\mathbb{E}\left[S|N' = n', M = m, N = n\right] = \sum_{i=1}^{m+n} c(n'+i),
\end{equation}
and 
\begin{equation}
\begin{aligned}
&\mathbb{P}[N'=n',M=m,N=n]=\mathbb{P}[N'=n'] \mathbb{P}[M=m]\mathbb{P}[N=n] \\
&=(1-p_c)p^{n'-1}_c (1-p_s)p^{m-1}_s (1-p_c)p^{n-1}_c,
\end{aligned}
\end{equation}
as $N'$, $M$, $N$ are independent with each other.
%and 
%$\mathbb{P}[N'=n']=(1-p_c)p^{n'-1}_c$, $\mathbb{P}[M=m]=(1-p_s)p^{m-1}_s$, and $\mathbb{P}[N=n]=(1-p_c)p^{n-1}_c$.
Let $p_0 \triangleq \max\{p_s,p_c\}$. We have
\begin{align}
&\mathbb{E}\left[S\right] 
\leq  \kappa \sum_{n'=1}^{\infty}\sum_{m=1}^{\infty}\sum_{n=1}^{\infty}  \sum_{i=1}^{m+n} c(n'+i) p^{n'+m+n}_0\\
%&<  \kappa \sum_{n'=1}^{\infty}\sum_{m=1}^{\infty}\sum_{n=1}^{\infty}  (m+n) c(n'+m+n) p^{n'+m+n}_0\\
&<  \kappa \sum_{n'=1}^{\infty}\sum_{m=1}^{\infty}\sum_{n=1}^{\infty}  (n'+m+n) c(n'+m+n) p^{n'+m+n}_0\\ \label{partition}
&<  \kappa \sum_{i=1}^{\infty} i^4 c(i) p^{i}_0.
\end{align}
where $\kappa =(1-p_c)p^{-1}_c (1-p_s)p^{-1}_s (1-p_c)p^{-1}_c$, and \eqref{partition} is due to the fact that the number of possible partition of $(n'+m+n)$ into three parts is less than $(n'+m+n)^3$.
Since there always exists $p'_0>p_0$ and $n$ such that $i^4p^i_0<(p'_0)^i, \forall i>n$, $\sum_{i=1}^{\infty} i^4 c(i) p^{i}_0<\infty$ if $\sum_{i=1}^{\infty} c(i) (p'_0)^i <\infty$.
Using the result that 
$\sum_{j=1}^{\infty} (p'_0)^j c(j) < \infty$ iff $p'_0 \rho^2(\mathbf{A})<1$ in \cite{schenato2007foundations} and \cite{schenato2008optimal}, $\sum_{i=1}^{\infty} i^4 c(i) p^{i}_0<\infty$ if $p_0 \rho^2(\mathbf{A})<1$, completing the proof.

\section*{Appendix C: Proof of Theorem~\ref{theorem:naive}}
The necessity and sufficiency are proved as follows.

\subsection{Sufficiency}
Similar to the proof of Theorem~\ref{theorem:existence}, 
we need to define the control cycle of the naive policy and then calculate the average cost.

Different from the Proof of Theorem~\ref{theorem:existence}, the control cycle is defined as the time slots after a effective control cycle until the end of the following effective control cycle. Here, the effective control cycle is the sequence of time slots starting from a sensor's successful transmission and ending at a controller's successful transmission, where there is no successful transmissions in between. In other words, in an effective control cycle, the sensor's measurement at the beginning of the cycle will be utilized for generating a control command, which will be implemented on the plant by the end of the cycle.
The control cycle and the effective control cycle are illustrated~as
\begin{equation}\label{process2}
\{\cdots, \overbrace{\underbrace{\text{s} \check{\text{c}} \check{\text{s}} \text{c} \cdots \text{s} \text{c}}_{m'},\hspace{-0.3cm}\underbrace{\underbrace{\check{\text{s}} \text{c} \text{s} \text{c} \cdots \text{s} \check{\text{c}}}_{n'},}_{\hspace{-1cm}\text{effective control cycle } (t-1)}}^{\text{control cycle } (t-1)}
 \overbrace{\underbrace{\text{s} \text{c} \text{s} \check{\text{c}} \cdots \text{s} \text{c}}_{m},\hspace{-0.2cm}\underbrace{\underbrace{\check{\text{s}} \text{c} \text{s} \text{c} \cdots \text{s} \check{\text{c}}}_{n},}_{\hspace{-0.5cm}\text{effective control cycle } t}}^{\text{control cycle } t}\cdots \}
\end{equation}
where $n$ and $l=m+n$ are the number of time slots of an effective control cycle and a control cycle, respectively, and $\check{\text{s}}$ and $\check{\text{c}}$ denotes a successful sensor's transmission and controller's transmission, respectively.
Note that $m$ and $n$ are even numbers.

Similar to the proof of Theorem~\ref{theorem:existence}, 
$S$ and $L\triangleq M+N$ denote the sum cost of the plant and the number of transmissions in a control cycle, respectively. Also, $S$ and $L$ of the sequence of control cycles can be treated as ergodic Markov chains, i.e., $\{\cdots, S_t,S_{t+1},\cdots\}$ and  $\{\cdots, L_t,L_{t+1},\cdots\}$, where $t$ is the control-cycle index.
%We use $N'$ to denote the number of consecutive downlink transmissions before the current control cycle, which follows the same distribution of $N$.
Due to the ergodicity of $\{S_t\}$ and $\{L_t\}$, the average cost in \eqref{metric} can be rewritten~as \eqref{J}, where
\begin{align}
\mathbb{E}\left[S\right] &= \sum_{\frac{n'}{2}=1}^{\infty}\sum_{\frac{m}{2}=0}^{\infty}\sum_{\frac{n}{2}=1}^{\infty}\mathbb{E}\left[S|N' = n', M=m,N=n\right]\\ \notag
&\mathbb{P}[N'=n', M=m,N=n],\\
\mathbb{E}\left[L\right] &= \sum_{\frac{n'}{2}=1}^{\infty}\sum_{\frac{m}{2}=0}^{\infty}\sum_{\frac{n}{2}=1}^{\infty} l \mathbb{P}[N'=n', M=m,N=n],
\end{align}
where $M+N$ and $N$ are the length of the current control cycle and effective control cycle, respectively, and $N'$ is the length of the previous effective control cycle. It is clear that $N'$ is independent with $M$ and $N$.

Thus, the average cost $J$ is bounded if $\mathbb{E}\left[S\right]$ is.
From the naive policy and the definition of the control cycle and the effective control cycle, we can derive the following probability density functions as
\begin{equation}
\begin{aligned}
%&\mathbb{P}[N=n]=(1- p_c)p^{\frac{n}{2}-1}_c p^{\frac{n}{2}-1}_s,\ n=2,4,6,\cdots\\
&\mathbb{P}[M=m, N=n]\\
&=
\begin{cases}
(1-p_s) (1-p_c)p^{n/2-1}_c p^{n/2-1}_s,\ m=0,n=2,4,6\\
(1-p_s)\left(p^{m/2}_s + (1-p_s)\sum_{i=1}^{m/2}p^{i-1}_s p^{m/2+1-i}_c\right) \\
\hspace{1cm}\times(1-p_c)p^{n/2-1}_c p^{n/2-1}_s,\ m,n=2,4,6,\cdots
\end{cases}
\end{aligned}
\end{equation}
and thus
\begin{equation}
\begin{aligned}
\mathbb{P}[N'=n']
&=
\mathbb{P}[N=n']=\sum_{m=0}^{\infty}\mathbb{P}[M=m, N=n']\\
&=(1-p_s p_c)p^{\frac{n'}{2}-1}_c p^{\frac{n'}{2}-1}_s,\ n'=2,4,6,\cdots
\end{aligned}
\end{equation}
Then, it can be proved that 
\begin{equation}
\begin{aligned}
&\mathbb{P}[N'=n'] \leq \kappa_1 p^{n'}_0,\ \forall n'=2,4,6,\cdots\\
&\mathbb{P}[M=\!m, N=\!n] \!\leq\! \kappa_2 (1+m/2) p^{m/2}_0 p^{n}_0, \forall m\!=\!0,2,4,\cdots\!,\\
&\hspace{6.4cm}n\!=\!2,4,6,\!\cdots
\end{aligned}
\end{equation}
where $p_0 = \max \{p_s,p_c\}$, $\kappa_1= (1-p_s p_c)p_c^{-1} p_s^{-1}$, and $\kappa_2= (1-p_s)(1-p_c)p_c^{-1} p_s^{-1}$.

%and 
%\begin{equation}
%\begin{aligned}
%&\mathbb{P}[N'=n',M=m,N=n]=\mathbb{P}[N'=n'] \mathbb{P}[M=m]\mathbb{P}[N=n] \\
%&=(1-p_c)p^{n'-1}_c (1-p_s)p^{m-1}_s (1-p_c)p^{n-1}_c,
%\end{aligned}
%\end{equation}

%and 
%$\mathbb{P}[N'=n']=(1-p_c)p^{n'-1}_c$, $\mathbb{P}[M=m]=(1-p_s)p^{m-1}_s$, and $\mathbb{P}[N=n]=(1-p_c)p^{n-1}_c$.
Since $
\mathbb{E}\left[S|N' = n', M = m, N = n\right]$ $= \sum_{i=1}^{m+n} c(n'+i)$, we have
\begin{align}\notag
&\mathbb{E}\left[S\right] = \sum_{\frac{n'}{2}=1}^{\infty}\sum_{\frac{m}{2}=0}^{\infty}\sum_{\frac{n}{2}=1}^{\infty}\mathbb{E}\left[S|N' = n', M=m,N=n\right]\\
&\hspace{3.2cm}
\mathbb{P}[N'=n']\mathbb{P}[M=m,N=n],\\
&\leq  \kappa_1 \kappa_2 \sum_{\frac{n'}{2}=1}^{\infty}\sum_{\frac{m}{2}=0}^{\infty}\sum_{\frac{n}{2}=1}^{\infty}  \sum_{i=1}^{m+n} c(n'+i) (1+\frac{m}{2}) p^{n'+\frac{m}{2}+n}_0\\
%&<  \kappa \sum_{n'=1}^{\infty}\sum_{m=1}^{\infty}\sum_{n=1}^{\infty}  (m+n) c(n'+m+n) p^{n'+m+n}_0\\
&<  \kappa_1 \kappa_2 \sum_{\frac{n'}{2}=1}^{\infty}\sum_{\frac{m}{2}=0}^{\infty}\sum_{\frac{n}{2}=1}^{\infty}  (n'+m+n)^2 c(n'+m+n) p^{\frac{n'+m+n}{2}}_0\\
&< 4 \kappa_1 \kappa_2 \sum_{i=2}^{\infty}  i^5 c(2 i) p^{i}_0.
\end{align}
Since there always exists $p'_0>p_0$ and $\bar{n}$ such that $i^5p^i_0<(p'_0)^i, \forall i>\bar{n}$, $\sum_{i=2}^{\infty} i^5 c(2i) p^{i}_0<\infty$ if $\sum_{i=2}^{\infty} c(2i) (p'_0)^i <\infty$.
Also, we have $\sum_{i=2}^{\infty} c(2i) (p'_0)^i < \sum_{i=1}^{\infty} c(i) \sqrt{p'_0}^i$.
Using the result that 
$\sum_{j=1}^{\infty} \sqrt{p'_0}^j c(j) < \infty$ iff $\sqrt{p'_0} \rho^2(\mathbf{A})<1$ in \cite{schenato2007foundations} and \cite{schenato2008optimal}, $\sum_{i=2}^{\infty} i^5 c(2i) p^{i}_0<\infty$ if $\sqrt{p_0} \rho^2(\mathbf{A})<1$, completing the proof of sufficiency.

\subsection{Necessity}
To prove the necessity, we consider two ideal cases: the sensor's transmission is perfect, i.e., $p_s=0$, and the controller's transmission is perfect, i.e., $p_c=0$. In these cases, the stability conditions are the necessary condition that the plant can be stabilized by the naive policy. The proof requires the analysis of average cost of control cycles which follows similar steps in the proof of sufficiency. 
Since the average cost $J = \frac{\mathbb{E}\left[S\right]}{\mathbb{E}\left[L\right]}$ and $\mathbb{E}\left[L\right]$ is bounded straightforwardly, we only need to prove the necessary condition that the average sum cost of a control cycle is bounded, i.e., $\mathbb{E}\left[S\right]$.

In the ideal cases, we have
\begin{equation}
\mathbb{E}\left[S\right] = 
\begin{cases}
(1-p_c) \sum_{j=1}^{\infty} \sum_{i=2}^{2j+1} c(i) p^{j-1}_c , & p_s=0\\
(1-p_s) \sum_{j=1}^{\infty} \sum_{i=2}^{2j+1} c(i) p^{j-1}_s , & p_c=0\\
\end{cases}
\end{equation}
Therefore, if $\mathbb{E}\left[S\right]$ is bounded, we have
\begin{equation}
\begin{aligned}
&\sum_{i=1}^{\infty}c(2i)p^i_c<\infty,
\sum_{i=1}^{\infty}c(2i+1)p^i_c<\infty\\
&\sum_{i=1}^{\infty}c(2i)p^i_s<\infty,
\sum_{i=1}^{\infty}c(2i+1)p^i_s<\infty\\
\end{aligned}
\end{equation}
and hence
\begin{equation}
\sum_{i=1}^{\infty}c(i) \sqrt{p_c}^i<\infty,\sum_{i=1}^{\infty}c(i) \sqrt{p_s}^i<\infty.
\end{equation}
Using the result that 
$\sum_{j=1}^{\infty} \sqrt{p_0}^j c(j) < \infty$ iff $\sqrt{p_0} \rho^2(\mathbf{A})<1$ in \cite{schenato2007foundations} and \cite{schenato2008optimal},
the necessary condition that the average cost inducted by the naive policy is bounded, is $\sqrt{p_s}\rho^2(\mathbf{A})<1$ and  $\sqrt{p_c}\rho^2(\mathbf{A})<1$, completing the proof of necessity.

\setcounter{subsection}{0}
\section*{Appendix D: Proof of Theorem~\ref{theorem:existence2}}
The necessity and sufficiency are proved as follows.

\subsection{Sufficiency}
We construct a persistent-scheduling-like policy including three phases: 1) the sensor's transmission is consecutively scheduled until it is successful, and then 2) the controller's transmission is consecutively scheduled until a successful transmission, and then 3) none of the sensor nor the controller is scheduled for transmission in the following $v-1$ time slots, i.e., all the commands contained in the successfully received control packet will be implemented by the actuator, and then phase 1) and so on.

Then, following the similar steps of the proof of Theorem~\ref{theorem:existence}, it can be proved that the persistent-scheduling-like policy stabilizes the plant if \eqref{condition} holds.

\subsection{Necessity}
The proof is conducted by considering two virtual cases: 
1) the sensor's transmission is continuously scheduled, while there is a virtual control input $\mathbf{u}_k$ at each time slot that ideally resets $\mathbf{x}_k$ to $\mathbf{0}$ if the sensor's transmission is successful at $k$, and is $\mathbf{0}$ otherwise;
2) the controller's transmission is continuously scheduled, while 
the controller applies a virtual estimator that has perfect estimation of the plant states in each time slots.

It can be readily proved that the two virtual cases result in lower average costs than any feasible uplink-downlink scheduling policy. Then, following the similar steps in the proof of Theorem~\ref{theorem:existence} and~\ref{theorem:naive}, it can be shown that if the average cost of case 1) is bounded, $p_s<1/\rho^2(\mathbf{A})$ must be satisfied, and if the average cost of case 2) is bounded, $p_c<1/\rho^2(\mathbf{A})$ must be satisfied, completing the proof of necessity.
    
    \balance
    
	\ifCLASSOPTIONcaptionsoff
	\newpage
	\fi

% Generated by IEEEtran.bst, version: 1.14 (2015/08/26)

\end{document}